\documentclass[conference]{IEEEtran}
\IEEEoverridecommandlockouts
\usepackage[
    backend=biber,
    maxnames=3,
	maxcitenames=2,
    sorting=none,
    style=numeric-comp,
    doi=false,
    isbn=false,
    url=false,
]{biblatex}
\AtBeginBibliography{\normalsize}
\usepackage{graphicx}
\usepackage{amsmath,amssymb,amsfonts}
\usepackage{tikz}
\usepackage{multirow}
\usepackage{tcolorbox}
\usepackage{pgfplots}
\usepackage{xcolor}
\usepackage{textcomp}
\tcbuselibrary{skins, breakable}
\usetikzlibrary{shapes.geometric, arrows, positioning,patterns}

\usepackage[inline]{enumitem}
\setlist[enumerate]{leftmargin=*}
\setlist[itemize]{leftmargin=*}
\usepackage{nicefrac}

\tikzstyle{startstop} = [rectangle, rounded corners, minimum width=3cm, minimum height=1cm,
    text centered, draw=black, fill=white]
\tikzstyle{process} = [rectangle, minimum width=3cm, minimum height=1cm, text centered,
    draw=black, fill=white]
\tikzstyle{decision} = [diamond, aspect=2, draw=black, fill=white, text centered, inner sep=1pt]
\tikzstyle{data} = [cylinder, shape border rotate=90, draw=black, aspect=0.4,
    minimum height=1cm, minimum width=2cm, text centered, fill=white]
\tikzstyle{arrow} = [thick,->,>=stealth]
\tikzstyle{dashedarrow} = [thick,dashed,->,>=stealth]

\RequirePackage[T1]{fontenc}

\usepackage{common} 
\usepackage{listings-rust}
\usepackage{svg}
\usepackage{wrapfig}

\newcommand{\promptfuzz}{\textsc{PromptFuzz}\xspace}
\newcommand{\randoop}{\textsc{Randoop}\xspace}
\newcommand{\mutest}{$\mu$\textsc{test}\xspace}
\newcommand{\evosuite}{EvoSuite\xspace}
\newcommand{\lisa}{\textsc{Lisa}\xspace}
\newcommand{\lisaapi}{\textsc{Lisa}$_{\text{API}}$\xspace}
\newcommand{\lisainv}{\textsc{Lisa}$_{\text{inv}}$\xspace}
\newcommand{\citywalk}{\textsc{CITYWALK}\xspace}
\newcommand{\ossfuzz}{OSS-Fuzz\xspace}
\newcommand{\ngram}{$n$-gram\xspace}
\newcommand{\nbench}{7\xspace}
\newcommand{\apiseq}{API-sequence\xspace}

\newcommand{\lisabench}{\textsc{Lisa-Bench}\xspace}

\newcommand{\acne}{\textsc{Acne}\xspace}

\newcommand{\gptmini}{\texttt{gpt-5-mini}\xspace}
\newcommand{\gptfull}{\texttt{gpt-5.1}\xspace}

\newcommand{\rev}[1]{#1}

\begin{document}
%
\title{LLM-Based Invariant Testing for Software Functional Bugs}

\newcommand{\internmark}{\textsuperscript{\ddag}}
\newcommand{\equalmark}{\textsuperscript{\S}}


\author{%
\IEEEauthorblockN{Ruogu Yang\internmark\equalmark}
\IEEEauthorblockA{\textit{Northeastern University}\\
yang.r4@northeastern.edu}
\and
\IEEEauthorblockN{Yifeng He\equalmark}
\IEEEauthorblockA{\textit{University of California, Davis}\\
yfhe@ucdavis.edu}
\and
\IEEEauthorblockN{Yundi Xu}
\IEEEauthorblockA{\textit{University of California, Davis}\\
xydxu@ucdavis.edu}
\and
\IEEEauthorblockN{Yuqing Wei}
\IEEEauthorblockA{\textit{Southern University of Science and Technology}\\
12311043@mail.sustech.edu.cn}
\and
\IEEEauthorblockN{Hao Chen}
\IEEEauthorblockA{\textit{The University of Hong Kong}\\
chenho@hku.hk}
\thanks{\internmark Work done while interning at UC Davis. \equalmark Equal contribution.}
\thanks{This is the authors' extended version of the paper accepted at ISSRE 2026.}
}

\maketitle

\begin{abstract}
	Manually writing unit tests to uncover functional bugs in software libraries
	is not only time-consuming but also requires a deep understanding of the intended semantics of the APIs.
	Heuristic-based test generation methods suffer from low usability because they cannot reason about program semantics
	or interpret source code and documentation as humans do.
	Traditional fuzzing techniques like \ossfuzz often rely on crashes to detect bugs, but functional bugs do not always cause crashes.
	To overcome these limitations,
	we present \lisa, a novel \underline{L}LM-based \underline{i}nvariant testing framework for \underline{s}oftw\underline{a}re functional bugs.
	\lisa iteratively generates API sequences and program invariants guided by API $n$-gram feedback,
	achieving higher bug-detection rates and \rev{competitive} code coverage compared with both fuzzing and prior LLM-based test generation approaches\rev{, and reporting each finding as a high-confidence bug candidate for developer confirmation}.
\end{abstract}


\begin{IEEEkeywords}
Large language models, software testing, test generation
\end{IEEEkeywords}

\section{Introduction}

Software defects (bugs) cause security vulnerabilities and unintended behaviors, and testing grows harder as software scales~\cite{fenton2000quantitative,chou2001os,zhu2025bugsbenchmarks}. Among automated techniques, coverage-guided grey-box fuzzing is the most widely adopted~\cite{serebryany2016libfuzzer,fioraldi2020afl_pp,bohme2017directed}: it mutates inputs to explore deep paths for crashes and hangs, and has uncovered vulnerabilities in thousands of real-world projects~\cite{serebryany2017ossfuzz}.

However, fuzzing primarily targets \emph{implementation} bugs that crash the program, which sanitizers can catch~\cite{serebryany2016libfuzzer}. \emph{Functional} bugs, by contrast, stem from incorrect logic and produce wrong results without crashing; they are hard for dynamic testing to detect and have been called ``machine un-auditable''~\cite{wang2024smartinv}, as catching them requires domain knowledge of the API's \emph{intended semantics} that current automated techniques lack.

Unit testing is the primary method developers use to identify functional bugs~\cite{khorikov2020unit}.
Unit testing requires developers to manually specify three elements: inputs, expected outputs, and the testing semantics that sequence the calls; we refer to the tested functions as the library's APIs and their invocation order as an \apiseq. Writing tests by hand is costly and often neglected~\cite{zhao2017ci,beller2015when}, while heuristic-based generators~\cite{pacheco2007randoop,fraser2011evosuite} cannot reason about API source or documentation, generalize poorly to new libraries, and yield low-diversity, low-coverage tests~\cite{shamshiri2018how,panichella2020revisiting}.

Prior automated approaches share common limitations. Coverage-guided fuzzers waste effort exploring shallow error-handling paths. Meanwhile, although several LLM-based unit test generation (LLM-UT) methods have been proposed recently~\cite{rao2023catlm,he2024unitsyn,he2025fuzzaug,zhang2025citywalk}, these approaches frequently produce one-shot tests with trivial oracles such as \texttt{assert(ptr != NULL)}~\cite{he2025fuzzaug,zhang2025citywalk}. End-to-end coding agents suffer from search-space explosion when planning, invoking tools, and generating tests jointly~\cite{xia2024agentless}.

Prior LLM-UT work has evolved from generating only assertions for simple functions~\cite{nie2023TeCo} to producing complete test functions and full test files~\cite{rao2023catlm,he2024unitsyn,he2025fuzzaug}, and typically evaluates them by pass rate and coverage~\cite{jain2025testgeneval}. \textcite{he2025fuzzaug} note that this setup lacks a reliable oracle, the long-standing \emph{oracle problem}~\cite{howden1978theoretical,weyuker1982nontestable,Barr2015oracle}: when a generated test fails, a reader cannot tell whether the LLM's predicted input--output pair is wrong or the API is, so these methods cannot reliably distinguish incorrect oracles from genuine faults. We define \lisa's invariant-based alternative and its scope in \autoref{sec:inv-def}.

To mitigate this oracle ambiguity,
we propose decoupling LLM-UT into two independent components:
\begin{enumerate*}
    \item sequences of API calls representing the testing semantics,
    \item assertions that verify execution results.
\end{enumerate*}
To maximize code coverage in the generated API sequences,
we follow prior work~\cite{pacheco2007randoop,lyu2024promptfuzz} by using feedback to guide the LLM in selecting diverse API combinations and permuting their orders.
Because we generate only API sequences rather than fuzzing drivers, as in \promptfuzz~\cite{lyu2024promptfuzz},
we introduce a novel $n$-gram API combination coverage as a guidance mechanism (\autoref{sec:n-gram}).
To improve the reliability of the generated assertions,
we relax the requirement for LLMs to predict exact input-output pairs
and instead ask them to infer program invariants~\cite{pei2023can,wang2024smartinv}.
We refer to this testing paradigm as invariant testing.
To construct valid unit tests from these components,
we propose chunk-invariant reasoning (\autoref{sec:chunk-inv-gen}), which partitions the finalized API sequence into multiple semantically coherent chunks and instructs the LLM to insert invariants at the end of each chunk as critical program points.
In this setting, invariants serve as \rev{oracles that are weaker \emph{as specifications} yet more robust \emph{as oracles}} than those required by formal verification\rev{ (two distinct axes, elaborated in \autoref{sec:inv-def})}, enabling us to check useful correctness properties without requiring precise input-output specifications.
By decoupling API-sequence exploration from invariant construction, this two-stage design diversifies generated tests and improves their practicality. It achieves \rev{competitive} code coverage and enables automated detection of functional bugs without precise input-output specifications.
\lisa provides an end-to-end automated pipeline for unit test generation. While finding unknown bugs is a significant benefit, the generated test suites are themselves valuable for regression testing. This invariant-based design mitigates, but does not fully solve, the oracle problem\rev{; accordingly, \lisa reports \emph{high-confidence bug candidates} that a developer confirms rather than acting as a fully autonomous bug detector}.

We present \lisa, a novel \underline{L}LM-based \underline{i}nvariant testing framework for \underline{s}oftw\underline{a}re functional bug detection.
\lisa generates unit test functions iteratively, guided by $n$-gram API coverage,
and inserts program invariants between API chunks to detect functional bugs in the target APIs.
We evaluated \lisa on \nbench real-world C/C++ libraries.
Compared with state-of-the-art fuzzing tools (\ossfuzz~\cite{serebryany2017ossfuzz}),
the unit tests generated by \lisa achieve higher average branch coverage at the library level.
Moreover, when evaluated on re-introduced historical functional bugs, \lisa detects \rev{nine} more bugs than the state-of-the-art LLM-based unit testing framework \citywalk~\cite{zhang2025citywalk}.
We make the following contributions:
\begin{itemize}
    \item We present \lisa, a novel LLM-based invariant testing 
                framework that detects functional bugs in software 
                libraries. To the best of our knowledge, \lisa is 
                the first \rev{to recast functional-bug detection as a \emph{decoupled} two-stage problem, feedback-guided API-sequence synthesis followed by documentation-grounded invariant insertion at chunk boundaries, so that the contribution is the decomposition mechanism rather than the integration of existing components}, enabling automated
                detection of functional bugs without precise 
                input-output specifications.
    \item We introduce chunk-invariant reasoning and $n$-gram API coverage feedback to decouple \apiseq exploration from invariant construction, enabling accurate and diverse unit tests.
    \item We present \lisabench, the first benchmark for evaluating functional bug detectors. 
        We empirically validate \lisa on \lisabench, and its generated tests achieve higher average branch coverage than existing fuzzing approaches and detect more functional bugs than LLM-UT methods.

\end{itemize}

\section{Background}

\subsection{Bugs in Software Supply-Chain} 

Bugs in the software supply chain (\ie, upstream libraries)
are generally categorized into two classes:
\begin{itemize}
	\item \emph{Implementation bugs}:
	      These occur when developers make mistakes in the code implementation,
	      leading to incorrect handling of types, memory, or system resources.
	      Implementation bugs typically exhibit universal rather than domain-specific patterns,
	      such as software crashes or hangs.
	      Owing to their universality, such bugs
	      can often be detected by static code analysis tools with predefined rules~\cite{bessey2010few},
	      or by automated dynamic testing techniques targeting behaviors like crashes, such as fuzzing~\cite{serebryany2016libfuzzer} and symbolic execution~\cite{cadar2008klee}.

	\item \emph{Functional bugs}:
	      These arise when developers make logical errors in the program.
	      When triggered, the API does not crash or terminate abnormally;
	      instead, it produces incorrect outputs that deviate from the expected behavior described in the documentation.
	      Detecting functional bugs in software libraries typically requires developers to manually craft
	      expected input-output pairs as test cases~\cite{khorikov2020unit},
	      or to apply formal verification techniques~\cite{leino2010dafny} to prove correctness.
\end{itemize}

Due to their context-specific nature,
testing for functional bugs is difficult to automate.
Manually crafting high-quality test cases requires developers or domain experts
to spend substantial time writing non-feature testing code.
Moreover, there is no direct way to optimize these tests for exploring deep program states,
except by investing additional effort in writing more tests.
Formal verification also requires manual setup of theorem provers,
and such tools are not as easily integrated into the software development lifecycle
(\ie, continuous development and continuous integration)
as fuzzing and unit testing.
A middle ground for detecting functional bugs automatically
is property-based testing~\cite{claessen2000quick_check},
which allows developers to specify the expected properties of APIs
and then generates random
(or even coverage-guided~\cite{lampropoulos2019coverage})
inputs to test these APIs against their specified properties.
However, no fully automated testing approach yet exists
that can detect functional bugs while remaining as easy to integrate
into modern development pipelines as fuzzing tools for implementation bugs.

\subsection{Program Invariants} 

\begin{figure}[th]
	\input{code/inv_in_llvm.tex}
\end{figure}

Program invariants are conditions in code that must hold
for the program to proceed to the next stage of execution~\cite{ernst2001daikon}.
As the name suggests,
these conditions remain invariant with respect to the program's state:
for all internal states and variable configurations, they must be satisfied.
In modern large-scale software systems and libraries,
developers often \emph{assert} such invariants at critical program points,
\ie, locations in the code where certain conditions must hold
for the function to produce correct outputs.

\autoref{lst:llvm_invariants} shows an example implementation from the LLVM~\cite{lattner2004llvm,llvm-project} library.
This example contains two critical program points,
where developers use assertions to ensure the desired properties of the API.
The first critical point occurs after initialization,
where developers use invariants to verify that all variables are successfully initialized
before proceeding to further optimization.
The second occurs after compiler optimization,
where developers ensure that the computed results satisfy necessary correctness properties.
In general, if the appropriate program invariants hold at their corresponding critical program points,
it is less likely that the program contains functional bugs.

\subsection{Unit Testing}

Unit testing is the primary method developers use to assess functional correctness during software development~\cite{khorikov2020unit}.
A unit test provides a measure of \emph{correctness}~\cite{dijkstra1972notes,khorikov2020unit},
verifying that an API produces the expected results for given inputs.
Otherwise, a logical error is present in the implementation.
Because upstream libraries in the software supply chain are not standalone applications deployed in production environments,
where testing often requires mocking,
their unit tests typically follow the classical style: arrange, act, and assert (AAA)~\cite{khorikov2020unit}.
Previous work on constructing unit test generation datasets~\cite{he2024unitsyn} has also adopted this paradigm.

\begin{figure}[t]
	\input{code/unit_test_example.tex}
\end{figure}

\autoref{lst:AAA_unit_test} shows two examples of classical AAA-style unit test functions.
Both test functions begin with the \emph{arrange} stage,
where variables, memories, and objects are initialized.
The tests then proceed to the \emph{act} stage,
where the APIs under test are invoked with specified inputs.
Together, the arrange and act stages provide the \emph{testing semantics}.
The examples in \autoref{lst:AAA_unit_test} demonstrate invocations of a single API for simplicity;
however, in practice, the act stage often involves multiple APIs composed in sequence~\cite{pacheco2007randoop}.
After invoking the APIs, the final \emph{assert} stage verifies whether the results match the expected outputs.
A failed assertion indicates the presence of a logical error in the code.

\section{Methodology}
\subsection{Overview}
\label{sec:lisa1}

\begin{figure*}[t!]
	\centering
	\includegraphics[width=.7\textwidth]{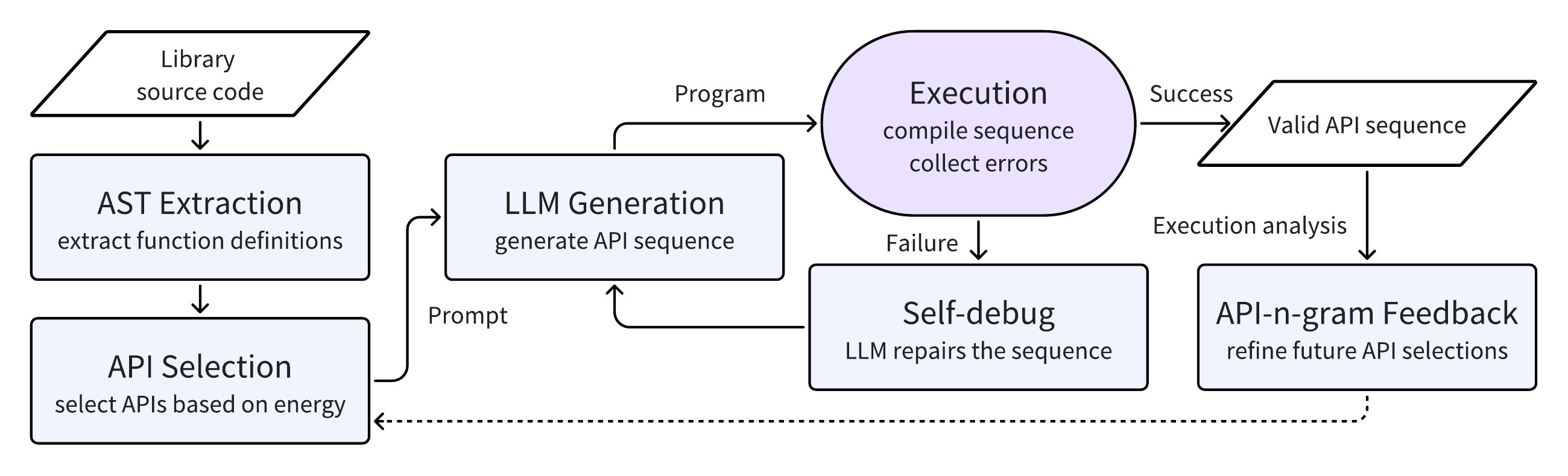}
	\caption{Overview of \lisa's iterative feedback-guided API sequence generation.}
	\label{fig:lisa1}
\end{figure*}

\lisa is a framework for automated unit test generation via LLM-based API reasoning, comprising two modules: \apiseq generation, which produces reasonable API invocation sequences for the library under test, and invariant insertion, which converts a structured API knowledge base into assertions inserted into the \apiseq. As depicted in \autoref{fig:lisa1}, \lisa first extracts library API and type definitions (\autoref{sec:API-extraction}), generates API sequences from selected APIs, project rules, and successful examples (\autoref{sec:API-seq-gen}), repairs erroneous sequences (\autoref{sec:error-detection-repair}), and iteratively refines them using API $n$-gram feedback~\cite{SANTOS2017Stepwise} (\autoref{sec:n-gram}). It then partitions each sequence into semantically meaningful segments~\cite{Ryan2024code}, converts knowledge-base information into invariants inserted per segment, and executes each segment to validate the invariants (repairing on failure), concatenating the validated segments into the final invariant-enriched test.

\subsection{API Information Extraction} \label{sec:API-extraction}

We build this component on top of \promptfuzz~\cite{lyu2024promptfuzz}, a prior LLM-guided API fuzzing framework for C/C++ libraries. \promptfuzz provides a practical front-end for library analysis and prompt construction, including Clang-based AST parsing and basic LLM invocation utilities, making it a suitable foundation for our pipeline. Importantly, \lisa reuses only these front-end components for API extraction and prompt orchestration; we do not use \promptfuzz's fuzzing loop, mutation strategy, or coverage-guided driver generation.

Following \promptfuzz, we use the abstract syntax tree (AST) to extract function and struct definitions from the target library.
\lisa extracts structured API metadata by statically analyzing header files using Clang's AST dump facility,
identifying function declarations, parameter and return types, and definitions of structs, enums, and typedef aliases
by parsing the generated JSON-formatted AST.
All types are then canonicalized into a standardized representation that preserves pointer mutability and array size information. This metadata is serialized into structured schemas that serve as the foundation for API-sequence generation, not fuzzing-driver synthesis.

\subsection{API Sequence Generation} \label{sec:API-seq-gen}
\lisa adopts a self-adaptive few-shot approach~\cite{wan-etal-2023-better} to generate API sequences using a generative large language model (LLM),
leveraging its capability to understand and produce valid code for open-source libraries.
We use \gptmini \rev{for its favorable cost and latency in this exploration-heavy phase, and switch to the stronger \gptfull for \lisainv (\autoref{sec:inv-gen}), where invariant reasoning benefits from a more capable model. Because \lisa's gains stem from its architecture rather than the base model, as our Vanilla-LLM ablation indicates (\autoref{tab:coverage_comparison}), these component gains are largely model-agnostic, and \lisa can be paired with a stronger model when the budget allows}.
To generate syntactically and logically correct code,
we construct a comprehensive prompt (\autoref{fig:prompt_template} in \autoref{sec:prompt-templates}) that restricts the search space using ground-truth \texttt{Library Context}, \texttt{Project Rules} (e.g., closing opened handles), and specific \texttt{API Combinations},
while \texttt{Code Requirements} enforce a straight-line execution policy.
Augmented by \texttt{Successful Examples} and a structured \texttt{Execution Schema (Initialize $\rightarrow$ Cleanup)},
the prompt guides \lisa to generate deterministic, high-quality \apiseq that reflect a complete usage life-cycle
within the library's actual implementation boundaries.

\subsection{Error Detection and Repair} \label{sec:error-detection-repair}
After generating the default number of API sequences,
\lisa attempts to identify and correct errors in these sequences in two steps: detection and repair.

\subsubsection{Detection}
\lisa applies two filters to validate \apiseq{s}. \emph{Static} checks use \texttt{clang} to compile and link the program, rejecting any sequence that fails. \emph{Dynamic} checks run the program in an isolated Docker sandbox with a 30-second timeout and catch segmentation faults, assertion failures, unhandled exceptions, and hangs. Recurring error patterns feed back into project-specific rules for later generation rounds; for \texttt{zlib}, for example, repeated segfaults from calls to \texttt{inflate()} on uninitialized streams caused \lisa to add the rule ``ensure the stream is initialized with \texttt{inflateInit()} before use,'' which the ablation study (\autoref{tab:ablation_lisa1}) shows is one of the recurring rules whose removal degrades execution success on stateful libraries. The \texttt{Successful Examples} slot in the prompt (\autoref{sec:API-seq-gen}) is populated the same way: \lisa retains sequences that compiled and executed in previous rounds and injects them as in-context exemplars for the next round, so the few-shot pool grows automatically.

\subsubsection{Repair}
Rather than discarding erroneous seeds, \lisa performs automatic repair by invoking a large language model (LLM) with an error-guided prompt.
When the generated API sequence fails either during compilation or execution, \lisa extracts the corresponding error message, including the error type, code, and diagnostic details
to construct a verbal feedback~\cite{shinn2023reflexion,wang2024repogenreflex}.
These fields are substituted into our repair template (\autoref{fig:repair} in \autoref{sec:prompt-templates}),
which instructs the LLM to regenerate a fixed version of the same program without altering its overall logic or structure.
The template explicitly restricts the model from redefining the \texttt{main} function, changing function names or parameters, or introducing control-flow constructs such as loops or branches.
This constraint ensures the semantic equivalence between the repaired code and the original while correcting low-level implementation errors (e.g., undeclared variables, mismatched types, or segmentation faults).

If the error type indicates an execution failure without detailed diagnostics, \lisa interprets it as a potential segmentation fault and prompts the LLM to strengthen pointer and memory safety.
\lisa recompiles and re-executes each repaired candidate to verify the fix.
If the repaired sequence is still invalid, \lisa abandons the seed to prevent unnecessary API invocations and computation overhead.
\rev{In practice, this \emph{sequence-repair} loop performs at most one repair attempt per erroneous program: a single LLM-guided fix recovers most compilation and runtime errors, whereas further attempts on a still-failing sequence rarely succeed and only add cost, so the seed is discarded rather than retried.}

\subsection{API N-Gram Feedback} \label{sec:n-gram}

To statistically capture the local co-occurrence patterns of API calls in real-world programs, we adopt an \emph{API-\ngram} model.
Similar to \ngram models in natural language processing, an API-\ngram represents a contiguous subsequence of $N$ API invocations extracted from program traces or source code.
For example, the sequence \texttt{\{open, read, close\}} constitutes an API-3-gram, while \texttt{\{malloc, memcpy, free, printf\}} forms an API-4-gram.
\rev{By modeling which APIs co-occur and in what order, the API-\ngram serves as a lightweight, data-driven dependency model: it captures usage and ordering dependencies among APIs without constructing an explicit API-dependency or interaction graph, as used by dependency-aware approaches such as Hopper~\cite{chen2023hopper} and \citywalk~\cite{zhang2025citywalk}. Integrating such richer dependency models to further improve sequence validity is a promising extension that we leave to future work.}
\rev{We set $N = 3$ for \lisa. \autoref{tab:ngram_sensitivity} reports a sensitivity analysis on \texttt{zlib} under an identical 3-hour budget: $N = 3$ attains the highest line and branch coverage and yields the most valid programs, whereas $N = 2$ under-constrains the sampled API combinations and $N = 4$ makes them too sparse to satisfy, lowering both validity and coverage. We therefore adopt $N = 3$ throughout; the full sweep is also available in our artifact.}

\begin{table}[t]
  \centering
  \footnotesize
  \caption{\rev{Sensitivity to the $n$-gram order $N$ on \texttt{zlib} (3-hour budget). \#Valid: successfully generated programs; coverage measured with \texttt{llvm-cov}.}}
  \label{tab:ngram_sensitivity}
  \begin{tabular}{cccc}
    \toprule
    $N$ & \#Valid & Line Cov. & Branch Cov. \\
    \midrule
    \rev{2} & \rev{512} & \rev{71.48\%} & \rev{59.22\%} \\
    \rev{\textbf{3}} & \rev{\textbf{818}} & \rev{\textbf{73.01\%}} & \rev{\textbf{60.01\%}} \\
    \rev{4} & \rev{705} & \rev{69.89\%} & \rev{57.02\%} \\
    \bottomrule
  \end{tabular}
\end{table}

\subsubsection{API Scheduling} \label{sec:api-scheduling}
Let \(E\) denote the energy of an API \(a\).
We assign a baseline energy to each API \(a_i\) to initialize the feedback loop:
\(E(a_i) = 1\).
When a successful program is generated, we extract all consecutive API 3-grams (triples) from its execution trace
$\mathcal{T} = \{(a_i, a_{i+1}, a_{i+2}) \mid i = 1, 2, \ldots, L-2\}$,
where $L$ is the length of the API call sequence.
For each newly discovered 3-gram $(a_i, a_j, a_k) \in \mathcal{T}$ that has not been observed before, we update the energies of all three constituent APIs:
\begin{equation} \label{eqn:energy-update}
	E(a) \gets E(a) + 1, \text{for each } a \in (a_i, a_j, a_k)
\end{equation}
This additive reward scheme ensures that APIs appearing in successful execution patterns accumulate higher energy over time, with each API starting from the same baseline value.

\subsubsection{Adaptive Condensed Normalization of Energy (\acne)}
We assign energy to successive $n$-grams in the API call sequence. However, assigning sampling probabilities proportional to raw energy skews the distribution heavily: as sampling progresses, a few high-energy $n$-grams dominate the probabilities and cause premature convergence, leaving other API functions under-explored. To encourage continued exploration of low-energy, potentially under-explored API functions, we propose \acne, which dynamically adjusts energy values based on their distribution to balance exploration and exploitation.

\paragraph{Normalization}
Recall from \autoref{eqn:energy-update} that our energy assignment is additive: each time a successful $n$-gram is observed, the energies of its constituent APIs are incremented by 1. The first step of \acne normalizes the energy values into a probability-friendly range. Let $E_{\min}$ and $E_{\max}$ denote the minimum and maximum energy values among all API functions. We normalize the energy of a given API function $a$ with min-max normalization, and ensure every API function keeps a non-zero selection probability by adding a small constant $\varepsilon > 0$:
\begin{equation} \label{eqn:energy-normalization}
	\hat{E}(a) = \varepsilon + (1 - \varepsilon)\bar{E}(a), \quad \text{where }
	\bar{E}(a) = \frac{E(a) - E_{\min}}{E_{\max} - E_{\min}}.
\end{equation}
We use $\varepsilon = 0.01$ in our implementation, so $\hat{E}(a) \in [\varepsilon, 1]$. At the first iteration every API shares the same energy ($E_{\min} = E_{\max}$); in this degenerate case we skip normalization and sample APIs uniformly.

\begin{definition}[Condensation transformation] \label{def:condense}
	Let $I = [0, 1]$ denote the interval containing the (normalized) energies of all API functions. A function $f: I \to [0, \infty)$ is a \emph{condense transformation} if it satisfies:
	\begin{enumerate}
		\item Strictly monotone increasing (order-preserving): $\forall x, y \in I,\ x < y \implies f(x) < f(y)$.
		\item Strictly concave: for $\lambda \in (0, 1)$ and $\forall x, y \in I$ with $x \neq y$, $f(\lambda x + (1 - \lambda) y) > \lambda f(x) + (1 - \lambda) f(y)$.
		\item Endpoint preserving: $f(0) = 0$.
	\end{enumerate}
\end{definition}

\begin{proposition}[Range compression] \label{prop:range-compression}
	Any condense transformation $f$ compresses the range of input values: for any $x, y \in I$ with $0 < x < y \leq 1$, $\;1 < \nicefrac{f(y)}{f(x)} < \nicefrac{y}{x}$.
\end{proposition}
\noindent We prove \autoref{prop:range-compression} in \autoref{sec:proofs-of-condense-transformation}.

\paragraph{Condensation}
To mitigate the skewness of the energy distribution, we apply a condense transformation to the normalized energy values. Its properties ensure that
\begin{enumerate*}
	\item API functions with higher energy \emph{always} have a higher probability of being selected by \lisa, and
	\item the marginal effect of increasing energy decreases as energy grows.
\end{enumerate*}
A valid condense transformation always reduces the relative differences between energy values (\autoref{prop:range-compression}). In \lisa, we use \emph{power condensation}: for a given API function $a$, with a parameter $\alpha$ controlling the degree of condensation,
\begin{equation} \label{eqn:power-condense}
	C(a) = \hat{E}(a)^{\alpha}, \quad \alpha \in (0, 1).
\end{equation}
We show that power condensation is a valid condense transformation in \autoref{sec:proofs-of-condense-transformation}.

\paragraph{Distribution-aware adaptive condensation}
The parameter $\alpha$ controls the degree of condensation: a small $\alpha$ largely boosts the probabilities of low-energy API functions (encouraging exploration), whereas $\alpha$ close to $1$ makes the effect mild (favoring exploitation). A fixed $\alpha$ is suboptimal because different stages of \lisa require different exploration--exploitation trade-offs: early on we want to identify likely-useful API functions quickly (exploitation), and once such a set is identified and sufficiently tested we want to explore other potentially useful API functions (exploration). We quantify the skewness of the energy distribution at iteration $t$ with the coefficient of variation $s_t = \nicefrac{\sigma_t}{\mu_t}$, the ratio of the standard deviation $\sigma_t$ to the mean $\mu_t$ of the normalized energies at iteration $t$, and adaptively set
\begin{equation} \label{eqn:adaptive-condensation}
	\alpha_t = \alpha_{\min} + (1 - \alpha_{\min})\, e^{-s_t}.
\end{equation}
Thus, when the distribution is highly skewed (large $s_t$), $\alpha_t$ is automatically lowered to encourage exploration. Here $\alpha_{\min}$ controls the maximum degree of condensation; we set $\alpha_{\min} = 0.5$ (derived in \autoref{sec:choice-of-alpha-min}).

\paragraph{Sampling new API functions}
Finally, we use the condensed energy values to sample new API functions. At iteration $t$, the probability of selecting an API function $a$ is
\begin{equation} \label{eqn:sampling-prob}
	P_t(a) = \frac{C_t(a)}{\sum_{i=1}^{N} C_t(a_i)}, \quad \text{where }
	C_t(a) = \hat{E}_t(a)^{\alpha_t}.
\end{equation}
We use these probabilities to sample a subset of APIs as the \texttt{API Combinations} for \apiseq generation, thereby guiding the LLM's attention.

\subsection{\rev{The Invariants \lisa Generates}} \label{sec:inv-def}

\rev{An invariant in \lisa is a Boolean predicate $\varphi$ over program state that must hold every time execution reaches a designated program point~\cite{ernst2001daikon,pei2023can,wang2024smartinv}. Concretely, for a generated \apiseq partitioned into chunks $C_1\Vert\cdots\Vert C_N$ (\autoref{sec:chunk-inv-gen}), each invariant is a pair $\langle\varphi,\ell\rangle$ in which $\ell$ is the end of some chunk $C_i$ and $\varphi$ ranges over the variables, return values, and reachable struct fields that are live at $\ell$. \lisa compiles each such $\varphi$ into an executable \texttt{assert($\varphi$)}; the assertion is \emph{satisfied} on a given build when it never fires and \emph{violated} otherwise. An invariant is therefore a \emph{partial} specification: it constrains a property the program must preserve, such as a size relation, a state flag, or a return-code contract, without pinning down the exact input--output mapping. That partiality is what makes invariants usable as oracles for functional bugs, sidestepping the input--output ambiguity long recognized as the obstacle to testing \emph{non-testable} programs~\cite{weyuker1982nontestable,Barr2015oracle}.}

\rev{The predicates \lisa emits fall into four kinds of increasing semantic depth. \emph{Structural} invariants assert that a handle or state object is well-formed after a constructor-like call. \emph{Value-range} invariants keep a field within its documented domain. \emph{Relational and conservation} invariants link live values across a call: the bytes a call consumes plus the bytes that remain equal the input size, for example. \emph{Documentation-derived semantic contracts} encode behavioral intent that exact input--output oracles cannot easily express, such as the numeric identity of a checksum API or the guarantee that a destructor resets a state object. The first two kinds are largely structural; the last two carry the intent that lets \lisa detect silent logic errors. \lisa produces these predicates by prompting the LLM with the chunk code and the relevant API documentation from the knowledge base (\autoref{sec:api-contract-doc}), seeded by documentation-mined contracts and Daikon candidates; the algorithm and its verification-and-repair loop are the subject of \autoref{sec:chunk-inv-gen}. A documented \emph{read-only} guarantee, for instance, becomes \texttt{assert(memcmp(buf, saved, n) == 0)} on the relevant input buffer.}

\rev{A \lisa invariant is \emph{semantically valid} when it does not contradict the documented contract of the APIs in scope and holds on the reference build of the library, so that a later violation signals a behavioral deviation rather than an over-strong assertion. \lisa enforces the first condition during knowledge-base construction (\autoref{sec:api-contract-doc}) and the second automatically, by executing each candidate against the reference build before retaining it. This is a deliberately weaker guarantee than soundness in formal verification: \lisa does not prove $\varphi$ over all inputs, only that $\varphi$ is consistent with the documented contract and with observed reference behavior. In that sense \lisa's invariants are weaker as specifications, since they are partial and unproven, yet more robust as oracles, since asserting a stable property is less brittle than predicting an exact output~\cite{weyuker1982nontestable,Barr2015oracle,he2025fuzzaug}. Because \lisa does not prove its invariants, a violation is only a \emph{high-confidence bug candidate}: an assertion that survives the verification-and-repair loop yet still fails on the target build, and a developer confirms every candidate before we count it as a detected bug. We revisit the residual threats this raises in \autoref{sec:threats}.}

\subsection{Invariant Generation} \label{sec:inv-gen}

\begin{figure}[t]
	\centering
	\includegraphics[width=\columnwidth]{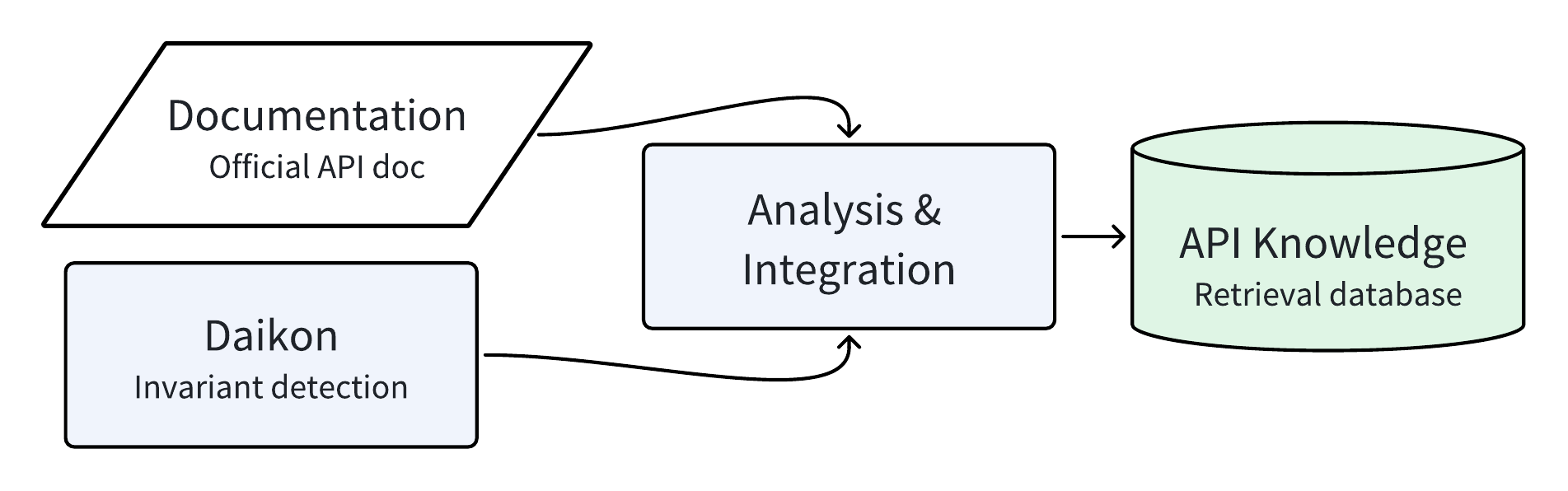}
	\caption{API knowledge database preparation.}
	\label{fig:api-d}
\end{figure}

Existing studies demonstrate that LLM-based invariant generation without structured guidance struggles to produce effective and reliable results~\cite{pei2023can}.
To overcome these limitations, \lisa first constructs an API contract documentation knowledge base to guide invariant generation, as shown in \autoref{fig:api-d} and detailed in \autoref{sec:api-contract-doc}.

As shown in \autoref{fig:invariant-pipeline}, \lisa adopts a segment-based invariant generation strategy to incrementally augment API-level programs with assertions.
The target program is first split into multiple segments.
For the first segment, the LLM takes the raw code as input and outputs the code augmented with inferred assertions.
For each subsequent segment, the input consists of the previously processed code and assertions, the current code segment, and relevant API documentation.
The model then generates assertions for the current segment by leveraging both the program context and external knowledge.
All processed segments are concatenated to produce the final unit tests enriched with invariants.

\begin{figure*}[t!]
	\centering
	\includegraphics[width=.92\textwidth]{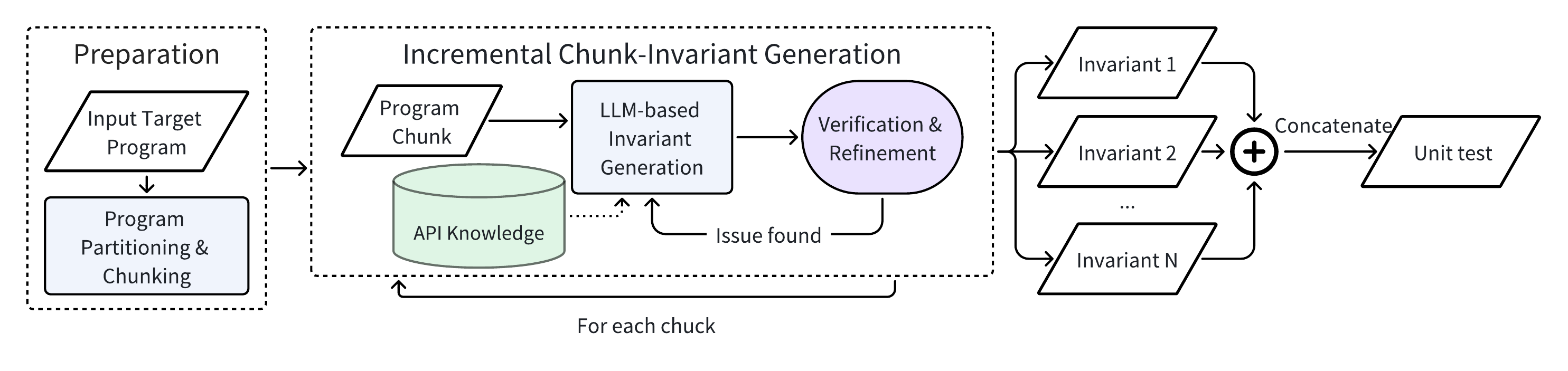}
	\caption{Overview of \lisa's incremental invariant generation pipeline.}
	\label{fig:invariant-pipeline}
\end{figure*}

Notably, LLM infers likely invariants based on the constructed API contract documentation.
When an execution failure occurs, \lisa starts a repair loop and retries up to five times to eliminate transient test bugs, such as mistakes in the generated test code or invariants.
If the failure persists after all retries, \lisa marks it as a high-confidence bug candidate, since repeated failures are unlikely to be caused by random or temporary test errors.
Nevertheless, because semantic reasoning is inherently complex, some persistent failures may still result from subtle test-side inaccuracies rather than actual implementation defects. 
These candidates are forwarded to human reviewers for final verification.

\subsubsection{Building API Knowledge Database}  \label{sec:api-contract-doc}
We build the API contract knowledge base through a hybrid approach:
\begin{enumerate*} 
    \item \textbf{Automated Extraction:} Daikon~\cite{ernst2000daikon} mines \emph{likely invariants} from executable API sequences, and an LLM extracts intended invariants from official API documentation~\cite{Xie2022DocTer} (e.g., ``read-only'' or ``no internal state modified'').
    \item \textbf{Manual Filtering:} We keep a Daikon candidate only if it does not contradict the documentation and does not cause assertion failures when added to the program.
    \item \textbf{Lightweight Manual Curation:} We supplement explicitly documented invariants missed by Daikon. This curation can also be performed incrementally on demand, rather than exhaustively documenting all APIs upfront. \rev{To quantify this manual cost, two annotators with 3+ years of C/C++ experience, after aligning on the annotation protocol, independently curated the knowledge-base entries for 20 APIs, resolving the few disagreements on which invariants to retain by discussion. They spent \num{5687} and \num{6231} seconds in total, i.e., \num{4.7} and \num{5.2} minutes per API respectively (overall mean \num{5.0} minutes per API), indicating that on-demand curation costs only a few minutes per API.}
\end{enumerate*}
We detail how Daikon-mined and documentation-derived invariants are reconciled in \autoref{sec:api-contract-doc-details}.

\subsubsection{Chunk-level Invariant Generation} \label{sec:chunk-inv-gen}

To keep reasoning local and stable, we partition each generated program into smaller chunks~\cite{liu-etal-2024-lost}: during API-sequence construction the LLM emits inline step markers (\texttt{// STEP1}, \texttt{// STEP2}, \dots) that we use to split the program by regular expressions. We then process the $N$ chunks $C_1,\dots,C_N$ left-to-right. For chunk $C_i$, the prompt contains the already-accepted code and assertions from $C_1\dots C_{i-1}$, the raw code of $C_i$, and the documentation for the APIs it uses; the LLM returns $C_i$ augmented with candidate assertions. This preserves prior state while keeping the context bounded.

\paragraph{Verification-and-repair loop}
Each set of candidate assertions for $C_i$ is immediately checked: we compile and execute $C_1 \Vert \cdots \Vert C_i$ and accept the assertions if no failure occurs. When failures arise, we enter a bounded repair loop \rev{(the \emph{invariant-repair} loop, up to $K{=}5$ attempts by default). In contrast to the single-attempt sequence-repair loop (\autoref{sec:error-detection-repair}), a failing assertion is often salvageable by weakening or correcting an over-strong invariant, so a few extra attempts recover many otherwise-discarded tests; we cap at $K{=}5$ because gains plateau beyond that}.
At each attempt, the LLM receives the failing assertion, triggering input (if available), and error trace, and is asked to weaken, repair, or remove the assertion. If the failure persists after $K$ attempts, we flag the seed as a potential unknown bug for human inspection, discard the assertion, and proceed. After all chunks are processed, the accumulated program forms the final invariant-enriched test. Compared with one-shot inference over the whole program, this incremental scheme shortens each prompt, localizes mistakes, and enables early detection and targeted repair of over-strong assertions via immediate execution feedback.

\section{Experimental Setup}
\subsection{Research Questions}
To evaluate \lisa, we pose the following three research questions. For convenience, we refer to the API sequence generation stage as \lisaapi and the invariant generation stage as \lisainv.
\begin{enumerate}[label=\textbf{RQ\arabic*}]
    \item \emph{How does \lisa perform in C/C++ unit test generation, and what advantages does it offer compared to existing approaches?}
          We evaluate \lisa's effectiveness in generating valid, meaningful C/C++ unit tests and compare it against state-of-the-art methods.
          
    \item \emph{How does each component impact the performance of \lisa?}
          \lisa consists of two stages, each comprising multiple components that contribute to its overall performance.
          We conduct ablation studies to understand the contribution of each component to \lisa's overall effectiveness.
          
    \item \emph{How effective is \lisa at detecting real-world functional bugs?}
          We evaluate \lisa's effectiveness in detecting bugs.
          Specifically, we focus on previously reported functional bugs in real-world projects, as these cases have been manually verified and confirmed by actual developers.
\end{enumerate}

\paragraph{Baseline}
To address these questions, we compare \lisa with a state-of-the-art unit test generation framework and relevant fuzzing techniques.
\begin{enumerate*}
\item \citywalk~\cite{zhang2025citywalk} is a recent C/C++ unit test generation framework that uses project-dependency awareness and language-specific knowledge to improve test reliability and quality.

\item \ossfuzz~\cite{serebryany2017ossfuzz} is Google's continuous fuzzing platform for open-source software. It combines coverage-guided fuzzing with continuous integration to detect bugs and security vulnerabilities in widely used libraries.
\end{enumerate*}
\citywalk represents the state-of-the-art in LLM-based unit test generation, making it a natural benchmark for assessing \lisa's ability to generate valid and high-quality unit tests.
\ossfuzz serves as a widely adopted industrial standard for fuzzing-based bug detection, providing a strong baseline for evaluating \lisa's effectiveness in code coverage and bug detection.
For \ossfuzz, we use the official fuzzing harnesses from the \ossfuzz GitHub repository and run \texttt{libFuzzer} for 3 hours on the same library version (identical commit) as \lisa. We measure branch and line coverage for all tools independently using \texttt{llvm-cov} with identical compiler flags, following Fuzzbench's recommendation~\cite{metzman2021fuzzbench} to use Clang source-based coverage rather than fuzzer-reported metrics. \autoref{sec:design-measurement-notes} records the remaining design and measurement details (harness handling, library-level coverage, and the comparison with \textsc{Hopper}).

\paragraph{\rev{Scope of baselines}}
\rev{We considered agentic SE workflows (SWE-agent~\cite{yang2024sweagent}, OpenHands~\cite{wang2025openhands}) and recent LLM-UT methods (IntUT~\cite{nan2025intut}) and excluded them for three reasons. \emph{Language}: SWE-agent and OpenHands target Python~\cite{jimenez2024swebench} and IntUT targets Java; none supports C/C++ test generation, whereas \lisa targets C/C++ library APIs. \emph{Toolchain coupling}: their pipelines rely on language-specific infrastructure (Python execution environments; IntUT's Java static analysis for test-intention derivation), so a fair comparison would require us to re-implement each system for C/C++ rather than run an existing tool. \emph{Task and oracle mismatch}: SWE-agent and OpenHands resolve a \emph{given} issue by patching a known defect rather than discovering unknown functional bugs, whereas IntUT relies on the exact-output oracle that \lisa explicitly avoids. Our Vanilla-LLM baseline and component ablations (\autoref{tab:coverage_comparison}, \autoref{tab:bug_finding}) address the orthogonal concern that \lisa's gains might come from its iterative generate-and-repair loop rather than the invariant oracle: they disable \lisa's components while keeping the same underlying LLM. We ground key parameters similarly: $N{=}3$ follows from the sensitivity analysis in \autoref{tab:ngram_sensitivity}, and the maximum condensation factor $\alpha_{\min}{=}0.5$ follows from the analytical derivation of \acne in \autoref{sec:n-gram} and \autoref{sec:choice-of-alpha-min}.}
\begin{table}[t]
    \centering
    \footnotesize
    \caption{%
        Benchmark for evaluation.
        LoC: line of code.
        \#APIs: number of public APIs.
    }
    \label{tab:benchmark}
    
\begin{tabular}{lcrrr}
	\toprule
	Library          & Commit ID        & LoC  & \#APIs    \\
	\midrule
	\textbf{zlib}    & \texttt{5a82f71} & 30k  & \num{87}  \\
	\textbf{libpng}  & \texttt{0f07f70} & 57k  & \num{246} \\
	\textbf{libpcap} & \texttt{d81c01c} & 58k  & \num{84}  \\
	\textbf{sqlite3} & \texttt{7efded5} & 413k & \num{289} \\
	\textbf{lcms}    & \texttt{8888d84} & 45k  & \num{286} \\
	\textbf{cJSON}   & \texttt{c859d25} & 10k  & \num{76}  \\
	\textbf{re2}     & \texttt{a4b2aee} & 28k  & \num{70}  \\
	\bottomrule
\end{tabular}

\end{table}

\subsection{\lisabench}
Existing benchmarks are insufficient for functional-bug-oriented test generation. Fuzzing benchmarks such as Magma~\cite{hazimeh2020magma} and FuzzBench~\cite{metzman2021fuzzbench} curate only crash-triggering bugs (e.g., overflows, use-after-free) that manifest as memory errors, not the silent logic errors functional bugs produce. LLM-UT benchmarks such as TestGenEval~\cite{jain2025testgeneval} rely on pass rate and coverage, which are insensitive to the \emph{oracle problem}: an assertion-free suite trivially achieves high pass rate and coverage yet detects nothing. Neither family, therefore, evaluates whether generated tests can expose functional defects.

To fill this gap, we introduce \lisabench, a benchmark specifically designed for evaluating functional-bug-oriented unit test generation for C/C++ libraries.
We select \nbench real-world open-source C/C++ projects from GitHub, \rev{five} of which include historical bugs.
The projects are selected according to four criteria: (1) the library exposes a public C/C++ API; (2) it builds and executes in our environment with reasonable effort; (3) it has sufficient functional complexity for invariant-based testing; and (4) historical bug reports or commits are available, enabling confirmed functional-bug cases. We also intentionally include projects from different application domains and codebase sizes to avoid overfitting the evaluation to a single library style.
\rev{As shown in \autoref{tab:benchmark}, these projects span diverse application domains: data compression (zlib~\cite{zlib}), PNG image processing (libpng~\cite{libpng}), network packet capture (libpcap~\cite{libpcap}), an embedded SQL database engine (sqlite3~\cite{sqlite3}), color management (lcms~\cite{lcms}), structured-format parsing (cJSON~\cite{cJSON}), and regular-expression matching (re2~\cite{re2}).}
Spanning small to very large code bases, these libraries let us curate historically reported, developer-confirmed functional bugs and directly measure whether generated tests detect real-world defects.

\subsection{Metrics} \label{sec:metric}

\subsubsection{Evaluating \lisaapi}
We evaluate \lisaapi using four standard metrics~\cite{schafer2024An,liu2025llmbasedunittestgeneration}: Compilation Success Rate (CSR) and Execution Success Rate (ESR) measure the fractions of generated \apiseq{s} that compile and run without runtime errors, respectively; Line Coverage (LC) and Branch Coverage (BC) are measured with \texttt{llvm-cov}.

\subsubsection{Evaluating \lisainv} \label{sec:eval-metrics}
We evaluate \lisainv using two metrics that reflect both its ability to generate assertions and their quality.

\paragraph{Average Unique Verifications Count (AUVC)}
AUVC measures the expected number of valid, unique verifications produced per input seed. Let $\mathcal{T}_{total}$ be the valid seeds from \lisaapi and $\mathcal{T}_{pass}\subseteq\mathcal{T}_{total}$ those that still compile and execute after invariant insertion; for $t\in\mathcal{T}_{pass}$, let $U(t)$ be its number of unique verifications. Then $\text{AUVC} = \sum_{t\in\mathcal{T}_{pass}} U(t)\,/\,|\mathcal{T}_{total}|$. Test cases whose invariants are invalid or violate intended behavior are excluded from $\mathcal{T}_{pass}$ and thus contribute zero, penalizing failed generation. \rev{Uniqueness is semantic, not syntactic: two annotators with 3+ years of C/C++ experience label each assertion, treating repeated queries over an evolving stream state as distinct while collapsing syntactically different but equivalent predicates such as \texttt{x >= 1} and \texttt{x > 0} for $x \in \mathbb{Z}$. Disagreements are resolved by discussion to consensus (full protocol in \autoref{sec:auvc-protocol}).}

\paragraph{Historical Bug Detection Rate (HBDR)}
We measure fault detection via \emph{historical bug re-introduction}: for \rev{five} libraries, we revert five previously fixed bugs in each (\rev{25} bugs in total). A bug counts as detected when the generated test passes on the patched branch but fails on the buggy one. \rev{To draw the five bugs for each library, we scan developer-fixed commits in reverse chronological order and keep the first five whose fix message flags a non-crashing functional defect and whose buggy revision reproduces deterministically on our toolchain; we skip commits that fail either filter and examine the next candidate. This protocol fixes the corpus rather than hand-picking it.} \rev{We prefer this design to mutation score, the alternative common in test-generation studies~\cite{Andrews2005Is}: synthetic mutants often diverge from the fault distribution of real software~\cite{Papadakis2018Are}, whereas a developer-verified historical fix bounds each bug's semantics and repair on evidence that a human maintainer already accepted.}
\section{Results and Analysis}

\subsection{Coverage and Success Rate}
\begin{table}[t]
    \centering
    \footnotesize
    \caption{Comparison of Compilation and Execution Success Rates between Vanilla LLM and \lisa.}
    \label{tab:successful_rate}
    
\begin{tabular}{lcccc}
    \toprule
    \multirow{2}{*}{\textbf{Library}} & \multicolumn{2}{c}{\textbf{Vanilla LLM}} & \multicolumn{2}{c}{\textbf{\lisa}}                                     \\
    \cmidrule(lr){2-3} \cmidrule(lr){4-5}
                                      & \textbf{Comp.}                           & \textbf{Exec.}                     & \textbf{Comp.}  & \textbf{Exec.}  \\
    \midrule
    \textbf{libpcap}                  & 80.0\%                                   & 45.2\%                             & \textbf{99.3\%} & \textbf{90.4\%} \\
    \textbf{lcms}                     & 92.0\%                                   & 30.6\%                             & \textbf{95.0\%} & \textbf{77.5\%} \\
    \textbf{zlib}                     & 94.3\%                                   & 92.6\%                             & \textbf{99.5\%} & \textbf{98.6\%} \\
    \textbf{sqlite3}                  & 88.3\%                                   & 28.6\%                             & \textbf{98.3\%} & \textbf{91.5\%} \\
    \textbf{re2}                      & 96.4\%                                   & 73.0\%                             & \textbf{100\%}  & \textbf{99.9\%} \\
    \textbf{cJSON}                    & 56.8\%                                   & 56.1\%                             & \textbf{100\%}  & \textbf{99.4\%} \\
    \textbf{libpng}                   & 65.1\%                                   & 38.2\%                             & \textbf{93.5\%} & \textbf{79.2\%} \\
    \midrule
    \textbf{average}                  & 81.8\%                                   & 52.0\%                             & \textbf{97.9\%} & \textbf{90.9\%} \\
    \bottomrule
\end{tabular}

\end{table}

To ensure a fair and rigorous comparison, we enforced a strict 3-hour time budget for all dynamic execution experiments,
including \lisa, the Vanilla LLM baseline, and \ossfuzz.
We present our results in \autoref{tab:successful_rate} and \autoref{tab:coverage_comparison}.
\lisa significantly outperforms both the Vanilla LLM and \ossfuzz across multiple dimensions.

\begin{table}[t]  
    \centering
    \small
    \caption{Code Coverage Comparison: \lisa vs. Vanilla LLM vs. \ossfuzz. All tools were run for 3 hours. \textbf{Bold} indicates the best performance.}
    \label{tab:coverage_comparison}
    \begin{adjustbox}{max width=\columnwidth}
        
\begin{tabular}{lcccccc}
    \toprule
    \multirow{2}{*}{\textbf{Library}} & \multicolumn{2}{c}{\textbf{Vanilla LLM}} & \multicolumn{2}{c}{\textbf{\ossfuzz}} & \multicolumn{2}{c}{\textbf{\lisa}}                                                          \\
    \cmidrule(lr){2-3} \cmidrule(lr){4-5} \cmidrule(lr){6-7}
                                      & \textbf{Line}                            & \textbf{Branch}                       & \textbf{Line}                      & \textbf{Branch}  & \textbf{Line}    & \textbf{Branch}  \\
    \midrule
    \textbf{libpcap}                  & 21.06\%                                  & 24.60\%                               & 25.79\%                            & 27.51\%          & \textbf{27.48\%} & \textbf{30.32\%} \\
    \textbf{lcms}                     & 30.28\%                                  & 20.82\%                               & 44.74\%                            & \textbf{34.44\%} & \textbf{49.18\%} & 27.59\%          \\
    \textbf{zlib}                     & 53.32\%                                  & 41.43\%                               & 65.99\%                            & 52.81\%          & \textbf{73.01\%} & \textbf{60.01\%} \\
    \textbf{sqlite3}                  & 14.40\%                                  & 10.54\%                               & \textbf{35.38\%}                   & \textbf{26.60\%} & 31.06\%          & 22.46\%          \\
    \textbf{re2}                      & 36.09\%                                  & 41.03\%                               & 29.77\%                            & 31.62\%          & \textbf{39.77\%} & \textbf{46.54\%} \\
    \textbf{cJSON}                    & 51.01\%                                  & 47.53\%                               & 42.28\%                            & 45.52\%          & \textbf{76.80\%} & \textbf{69.07\%} \\
    \textbf{libpng}                   & 13.01\%                                  & 8.92\%                                & \textbf{22.57\%}                   & \textbf{17.94\%} & 20.21\%          & 14.92\%          \\
    \midrule
    \textbf{\rev{average}}            & \rev{31.31\%}                            & \rev{27.84\%}                         & \rev{38.07\%}                      & \rev{33.78\%}    & \rev{\textbf{45.36\%}} & \rev{\textbf{38.70\%}} \\
    \bottomrule
\end{tabular}
 
    \end{adjustbox}
\end{table}

\paragraph{Superiority over Vanilla LLM} Our Vanilla LLM baseline disables all of \lisa's key components (feedback, repair, chunking, rules), retaining only basic correctness checks and a minimal prompt with API knowledge. The results demonstrate that \lisa's superior performance stems from its architectural design rather than the underlying LLM capability. \lisa achieves a 38\% higher average execution rate and more than doubles the line coverage in complex libraries such as \texttt{sqlite3} (31.06\% vs. 14.40\%). This confirms \lisa's $n$-gram feedback and repair mechanisms are essential for guiding LLM beyond shallow happy paths to explore deep, unseen code.

\paragraph{Comparison with Fuzzing} Compared to the industrial standard \ossfuzz, \lisa demonstrates remarkable efficiency, achieving higher coverage on the majority of libraries.
\rev{\lisa leads \ossfuzz in line coverage on five of the seven libraries, with gains ranging from +1.7\% on \texttt{libpcap} to +34.5\% on \texttt{cJSON} (e.g., +7.0\% on \texttt{zlib} and +10.0\% on \texttt{re2}).} For libraries requiring structured inputs (\texttt{cJSON}, \texttt{re2}) or complex state transitions (\texttt{zlib}), the LLM's inherent knowledge of syntax and state flows bypasses the grammar barriers and state-space traps that hinder stochastic fuzzers in short-duration runs. Conversely, \ossfuzz leads on the two libraries with extensive error-handling branches for malformed inputs (\texttt{sqlite3}, \texttt{libpng}), where input mutation excels; however, the coverage gained from these error paths rarely translates into the discovery of functional bugs.
\begin{figure}[t]
    \centering
        \begin{tikzpicture}
        \begin{axis}[
            ybar,
            width=\linewidth, 
            height=4.2cm,     
            bar width=8pt,    
            legend style={at={(0.5,0.98)}, anchor=north, legend columns=2, font=\scriptsize, /tikz/every even column/.append style={column sep=0.4cm}},
            symbolic x coords={zlib, re2, libpng, libpcap, lcms, sqlite, cJSON},
            xtick=data,
            nodes near coords, 
            nodes near coords style={font=\scriptsize, color=black}, 
            ymin=0, ymax=50,   
            ylabel={AUVC},
            ylabel near ticks,
            xlabel={Libraries},
            xlabel near ticks,
            label style={font=\scriptsize},
            x tick label style={font=\scriptsize},
            y tick label style={font=\scriptsize},
            grid=major,
            major grid style={dashed, gray!30}
        ]
            \addplot[fill=blue!60, draw=black] coordinates {
                (zlib, 34.3) (re2, 37.4) (libpng, 23.6) (libpcap, 21.9) (lcms, 17.7) (sqlite, 20.6) (cJSON, 44.3)
            };
            
            \addplot[fill=orange!50, draw=black, postaction={pattern=north east lines}] coordinates {
                (zlib, 7.4) (re2, 9.5) (libpng, 7.9) (libpcap, 8.2) (lcms, 7.6) (sqlite, 9.1) (cJSON, 11.8)
            };
            
            \legend{\lisa, \citywalk}
        \end{axis}
    \end{tikzpicture} 
    \caption{Comparison of AUVC between \lisa and \citywalk.}
    \label{fig:auvc_comparison} 
\end{figure}

\subsection{Quality of Generated Oracles: AUVC}
While code coverage measures how much code is executed, it does not reflect whether the execution is properly verified. 
We use AUVC (\autoref{sec:eval-metrics}) to evaluate the semantic richness of the generated tests.
We compare \lisa against \citywalk, the state-of-the-art LLM-based unit test generation framework.

\paragraph{High Density of Verifications} As illustrated in \autoref{fig:auvc_comparison}, \lisa demonstrates a substantial advantage in assertion density across all evaluated libraries. On average, \lisa generates \textbf{3.25$\times$} more unique verifications per test than \citywalk. For example, in \texttt{cJSON}, \lisa achieves an AUVC of 44.3 compared to \citywalk's 11.8. 

\paragraph{Semantic Depth} This performance gap stems from fundamental differences in generation strategies. \citywalk generates tests in a single pass ("one-test-per-method"), often producing simple checks such as \texttt{assert(ptr != NULL)} to ensure compilability. In contrast, \lisa's dual-stage design allows \lisainv to focus exclusively on logic verification. By leveraging the constructed API knowledge base, \lisa injects rich, state-aware invariants, for example, validating complex struct fields or return states, that \citywalk misses. This high AUVC confirms that \lisa not only executes code but rigorously verifies its correctness, directly contributing to its superior bug-finding capability.

\subsection{Ablation Study}
We perform an ablation study on \lisa's components, denoted \lisaapi and \lisainv. 
\subsubsection{Components of \lisaapi}
To rigorously assess the contribution of each component within \lisaapi, we conduct an ablation study on three representative libraries: \texttt{sqlite3}, \texttt{cJSON}, and \texttt{zlib}. \rev{We choose these three for three reasons: they span the complexity spectrum of our benchmark, ranging from a small structured-format parser (\texttt{cJSON}), through a stateful streaming codec (\texttt{zlib}), to a large and highly stateful database engine (\texttt{sqlite3}); they all belong to the bug-finding benchmark and thus reflect the settings where \lisa is ultimately evaluated; and limiting the ablation to three subjects keeps the cost of running every component variant tractable.} We derive four variants by selectively disabling specific modules from the full framework:
\begin{itemize*}
\item \textbf{w/o Rules:} Removes project-specific constraints and headers from the prompt.
\item \textbf{w/o Feedback:} Replaces the $n$-gram guided feedback mechanism with random selection.
\item \textbf{w/o Repair:} Disables the error detection and repair loop.
\item \textbf{w/o Examples:} Removes the few-shot successful examples from the prompt context.
\end{itemize*}
We evaluate these variants using the same metrics as RQ1. As shown in \autoref{tab:ablation_lisa1}, every component within \lisaapi is indispensable for effective API sequence generation.

For simpler libraries (\texttt{cJSON}, \texttt{zlib}), execution stays above 95\% even under ablation, yet the full framework still yields the best coverage (e.g., +7.43\% line coverage over random selection on \texttt{zlib}). For complex, state-sensitive libraries such as \texttt{sqlite3}, the components become critical: removing repair (\textit{w/o Repair}) drops execution to 51.4\% and line coverage to 25.53\%, as LLMs frequently hallucinate invalid states. The $n$-gram feedback is the single most important component for coverage: its removal yields the lowest coverage across all libraries (e.g., 21.40\% line coverage on \texttt{sqlite3}), since it otherwise steers generation toward under-explored API combinations.

\begin{table}[t]  
    \centering
    \small
    \caption{%
        Ablation study of \lisaapi components. 
        w/o Feedback: the variant with only random selection.
        Comp.: compile rate.
        Exec.: execution rate.
    }
    \label{tab:ablation_lisa1}
    \begin{adjustbox}{max width=\columnwidth}
    \begin{tabular}{llcccc}
    \toprule
    \textbf{Library} & \textbf{Variant}      & \textbf{Comp.} & \textbf{Exec.} & \textbf{Line Cov.} & \textbf{Branch Cov.} \\
    \midrule
    \multirow{5}{*}{\textbf{zlib}}
                     & w/o Rules             & 99.3\%              & 98.1\%              & 66.48\%            & 52.44\%              \\
                     & w/o Feedback & 99.0\%              & 98.4\%              & 65.58\%            & 52.48\%              \\
                     & w/o Examples          & 99.5\%              & 98.3\%              & 67.31\%            & 53.92\%              \\
                     & w/o Repair            & 98.8\%              & 95.7\%              & 67.58\%            & 54.30\%              \\
    \cmidrule{2-6}
                     & \textbf{\lisaapi}     & \textbf{99.5\%}     & \textbf{98.6\%}     & \textbf{73.01\%}   & \textbf{60.01\%}     \\
    \midrule
    \multirow{5}{*}{\textbf{cJSON}}
                     & w/o Rules             & 99.9\%              & 98.8\%              & 72.54\%            & 66.22\%              \\
                     & w/o Feedback & 100\%               & 98.5\%              & 71.80\%            & 63.95\%              \\
                     & w/o Examples          & 99.3\%              & 98.2\%              & 73.20\%            & 66.60\%              \\
                     & w/o Repair            & 99.8\%              & 97.0\%              & 72.72\%            & 65.37\%              \\
    \cmidrule{2-6}
                     & \textbf{\lisaapi}     & \textbf{100\%}      & \textbf{99.4\%}     & \textbf{76.80\%}   & \textbf{69.07\%}     \\
    \midrule
    \multirow{5}{*}{\textbf{sqlite3}}
                     & w/o Rules             & 97.6\%              & 83.2\%              & 26.88\%            & 19.37\%              \\
                     & w/o Feedback & 98.0\%              & 70.4\%              & 21.40\%            & 15.22\%              \\
                     & w/o Examples          & 96.2\%              & 74.8\%              & 27.27\%            & 19.68\%              \\
                     & w/o Repair            & 90.9\%              & 51.4\%              & 25.53\%            & 17.87\%              \\
    \cmidrule{2-6}
                     & \textbf{\lisaapi}     & \textbf{98.3\%}     & \textbf{91.5\%}     & \textbf{31.06\%}   & \textbf{22.46\%}     \\
    \bottomrule
\end{tabular}
 
    \end{adjustbox}
\end{table}

\subsubsection{Component of \lisainv}
\label{sec:ablation_lisa2}

Unlike \lisaapi, whose primary goal is to produce valid and diverse execution paths (measured by coverage), the contribution of the invariant generation component (\lisainv) is intrinsically linked to the ability to detect faults. A simple count of assertions (like AUVC) in ablation variants does not fully capture the \textit{correctness} or \textit{utility} of those assertions. 
Therefore, to provide a realistic evaluation, we defer the detailed ablation study of \lisainv (analyzing the impact of \textit{Repair}, \textit{Chunking}, and \textit{Knowledge}) to \textbf{RQ3}. There, we evaluate how each component contributes to functional-bug detection, including two knowledge-source variants: one using only Daikon outputs and one using only official documentation.

Although several \lisaapi components are implemented through prompts, \lisa is not simply a prompt-engineering variant. Its main contribution is the two-stage design: it first explores executable API sequences, then performs invariant generation as a separate stage grounded in documentation and chunk-level reasoning. This decomposition changes both the exploration target and how test oracles are constructed. The benefit is evident in the \lisaapi ablations (\autoref{tab:ablation_lisa1}) and in the bug-finding results, where removing \lisainv components consistently reduces detection effectiveness. These results suggest the gains come from the overall framework design, with prompting serving as merely one implementation mechanism.

\subsection{Bug-Finding Ability}
We evaluate the practical bug-finding ability of \lisa on \lisabench, targeting \rev{25} reproduced historical functional bugs \rev{across five libraries}.
To ensure a fair comparison, we aligned the execution constraints with the operational nature of each tool:
\begin{itemize*}
\item \ossfuzz: Guaranteed a continuous 6-hour fuzzing window to maximize mutation depth.
\item \citywalk: Allowed to run to completion, processing the entire set of focal methods without time truncation.
\item \lisa: Operated within a strict 3 hours pipeline for \lisaapi seed generation and generated final unit tests through \lisainv with the total runtime capped at 6 hours.
\end{itemize*}

\begin{table*}[t]  
    \centering
    \small
    \caption{Bug-finding results.}
    \label{tab:bug_finding}
    \scalebox{0.95}{
        \begin{tabular}{lc|cc|c|ccccc}
\toprule
\multirow{2}{*}{\textbf{Library}} & \multirow{2}{*}{\textbf{Total}} & \multicolumn{2}{c|}{\textbf{Baselines}} & \textbf{Ours} & \multicolumn{5}{c}{\textbf{Ablation of \lisainv}} \\
 & & \textbf{\citywalk} & \textbf{\ossfuzz} & \textbf{\lisa (Full)} & \textbf{w/o Repair} & \textbf{w/o Chunk} & \textbf{w/o Knowledge} & \textbf{Daikon Only} & \textbf{Doc Only}\\
\midrule
\textbf{cJSON}   & 5 & 1 & 0 & \textbf{3} & 1 & 2 & 1 & 2 & 2 \\
\textbf{lcms}    & 5 & 1 & 0 & \textbf{3} & 1 & 1 & 1 & 1 & 1 \\
\textbf{zlib}    & 5 & 0 & 1 & \textbf{2} & 1 & 0 & 0 & 1 & 0 \\
\textbf{sqlite3} & 5 & 1 & 0 & \textbf{2} & 1 & 1 & 1 & 1 & 1 \\
\rev{\textbf{libpng}}  & \rev{5} & \rev{0} & \rev{1} & \rev{\textbf{2}} & \rev{0} & \rev{0} & \rev{0} & \rev{0} & \rev{0} \\
\midrule
\textbf{Total}   & \rev{25} & \rev{3 (12\%)} & \rev{2 (8\%)} & \rev{\textbf{12 (48\%)}} & \rev{4 (16\%)} & \rev{4 (16\%)} & \rev{3 (12\%)} & \rev{5 (20\%)} & \rev{4 (16\%)} \\
\bottomrule
\end{tabular} 
    }
\end{table*}

As shown in \autoref{tab:bug_finding}, \lisa (Full) achieved a total detection rate of \rev{48\% (12/25)}, significantly outperforming both \citywalk (\rev{12\%}) and \ossfuzz (\rev{8\%}).

Fisher's exact test confirms that these differences are statistically significant: \lisa vs. \ossfuzz (\rev{$p = 0.0036$, odds ratio $= 10.6$}) and \lisa vs. \citywalk (\rev{$p = 0.0121$, odds ratio $= 6.77$}). Both $p$-values are below 0.05.

The baselines' poor performance reveals their structural limitations. \ossfuzz (\rev{2/25}) detects only crash-inducing faults (e.g., memory corruption) and misses functional logic errors that do not trigger runtime aborts. \rev{Tellingly, the single libpng defect that \ossfuzz caught but \lisa missed was a use-after-free that aborts at runtime rather than a silent logic error, underscoring the complementary scopes of the two tools.} \citywalk (\rev{3/25}), despite using LLMs, achieves a low detection rate because its "one-test-per-method" strategy prioritizes compilability through extensive mocking and isolated testing~\cite{zhang2025citywalk}. While this yields good coverage for focal functions, it cannot construct the complex, multi-step API sequences necessary to trigger state-related functional bugs.

In contrast, \lisa's dual-stage design enables superior detection. \lisaapi's feedback-guided exploration constructs diverse, valid API sequences that penetrate deep program states, while \lisainv injects invariants that reveal silent logic errors, helping \lisa find \rev{nine} more bugs than the best baseline.

The ablation study validates the necessity of \lisainv's three core components. First, repair is critical for test survivability: without it (\textit{w/o Repair}, 4 bugs), many tests terminate prematurely due to minor assertion failures or syntax errors and never reach the API calls that trigger bugs. Second, chunking is necessary for long sequences (\textit{w/o Chunk}, 4 bugs): without it, the model struggles to maintain attention across the entire program, missing intermediate states and failing to place assertions at key points. Third, the knowledge base provides essential semantic guidance (\textit{w/o Knowledge}, 3 bugs). Neither source alone suffices: Daikon-only introduces substantial noise (5 bugs), while documentation-only covers basic requirements but misses semantic constraints (4 bugs). Combining both gives \lisa broader semantic coverage and stronger filtering of spurious invariants.

\section{Discussion and Future Work}
\subsection{Case Study: Silent Logic Error in SQLite}
\begin{figure}[th]
    \scriptsize
	\input{code/sqlite}
\end{figure}
To illustrate \lisa's ability to find functional bugs, we analyze a logic regression in SQLite (commit d443f0a). As described in the commit message, a regression in the query optimizer caused the SQL expression \textbf{"0 OR 2"} to be erroneously evaluated as \textbf{2} (the integer value of the second operand) instead of 1 (Boolean TRUE). This is a classic logic error: the program executes valid CPU instructions and manages memory correctly, but the mathematical result is wrong. As shown in Listing~\ref{lst:sqlite_poc}, \lisa exposed this defect as follows:
\begin{itemize}
    \item \textbf{API Sequence Construction:} \lisa first generated a valid call sequence to prepare and execute the specific query \texttt{"SELECT 0 OR 2"} using the \texttt{sqlite3\_prepare\_v2} and \texttt{sqlite3\_step} APIs.
    \item \textbf{Invariant Assertion:} Crucially, instead of merely checking for a successful return code (e.g., \texttt{SQLITE\_ROW}), \lisa inferred the semantic invariant of the SQL operation. It generated the strict assertion \texttt{assert(val == 1)}, enforcing that the logical OR operation must yield a Boolean True (normalized to 1 by \lisa).
\end{itemize}
Consequently, although the buggy version returned \num{2} due to an incorrect bitwise optimization and would pass crash-based fuzzing, \lisa's test triggered an assertion failure.  

\subsection{\rev{Threats to Validity and Limitations}} \label{sec:threats}

\rev{We organize the residual threats along the standard three axes. \emph{Internal validity}: LLM non-determinism (\emph{Reproducibility}) and hyperparameter choices (\emph{Hyperparameter Settings}) could bias the measured effect; the variance study (\autoref{tab:variance}) shows the residual noise is well below \lisa's margin over the baselines. \emph{External validity}: the 25-bug, five-library, C/C++-only corpus limits generalization (\emph{Benchmark scale}); we mitigate this with a fixed selection protocol (\autoref{sec:metric}) and consistent per-library gains rather than aggregate-only claims. \emph{Construct validity}: a violated invariant is a \emph{high-confidence bug candidate}, not a proven defect (\emph{Final oracle verification}), and \lisa's recall is bounded by documentation completeness (\autoref{sec:inv-def}).}

\paragraph{\rev{Benchmark scale and representativeness}}
\rev{\lisabench currently comprises \num{25} developer-confirmed historical functional bugs spanning five libraries from distinct domains. While modest in absolute size, the set is assembled by a fixed protocol rather than hand-picked: each bug is drawn from a previously fixed historical commit, included only if it reproduces deterministically on our toolchain and exhibits a non-crashing functional symptom, and excluded otherwise. \lisa's advantage is consistent across the per-library breakdown (\autoref{tab:bug_finding}) rather than driven by any single library, and its margin over both baselines remains statistically significant (Fisher's exact test, $p<0.05$). Nonetheless, the absolute scale limits the precision of the point estimates, and evaluation on a larger and more varied bug corpus is an important direction for future work.}
\paragraph{\rev{Reproducibility and non-determinism}}
\rev{Although the underlying LLM is non-deterministic, this affects only \emph{which} tests \lisa generates, not their stability once generated: every generated test is a fixed C/C++ program that compiles and executes identically on every run, so \lisa does not produce flaky tests. Non-determinism during generation is further dampened by the two repair loops, which execute each candidate against the reference build and discard transient or unstable assertions, retaining only invariants that pass stably. To quantify the residual variance, we repeat the full pipeline three times on the three representative libraries (\autoref{tab:variance}). Compilation and execution success rates are highly stable, while line and branch coverage vary by only a few percentage points (standard deviation $\le 3.5$), far smaller than \lisa's margin over the baselines; in every run \lisa exceeds the Vanilla baseline on all three libraries and its ranking against \ossfuzz is unchanged. Bug detection is likewise stable (\autoref{tab:variance}, last column): all but one of the detected bugs recurs in every run. The main tables report one representative run from this set.}
\begin{table}[t]
  \centering
  \footnotesize
  \caption{\rev{Run-to-run variance of \lisa over three repetitions on three representative libraries. Comp./Exec.: compilation/execution success rate (mean$\pm$std, \%); Cov.\ in \%; Bugs: number of the 5 historical bugs detected in each of the three runs.}}
  \label{tab:variance}
  \begin{tabular}{lccccc}
    \toprule
    Library & Comp. & Exec. & Line Cov. & Branch Cov. & Bugs \\
    \midrule
    \rev{\texttt{zlib}}    & \rev{99.8$\pm$0.3} & \rev{98.3$\pm$0.5} & \rev{70.5$\pm$2.1} & \rev{57.1$\pm$2.5} & \rev{2, 2, 2} \\
    \rev{\texttt{cJSON}}   & \rev{100$\pm$0.0}  & \rev{98.9$\pm$0.5} & \rev{74.6$\pm$3.5} & \rev{67.5$\pm$2.9} & \rev{3, 3, 2} \\
    \rev{\texttt{sqlite3}} & \rev{98.5$\pm$0.9} & \rev{90.0$\pm$3.3} & \rev{31.4$\pm$1.8} & \rev{22.8$\pm$1.5} & \rev{2, 2, 2} \\
    \bottomrule
  \end{tabular}
\end{table}
\paragraph{Hyperparameter Settings}
The choice of hyperparameters, such as the value of $n$ in $n$-gram and the repair attempt limit (5), along with the selected LLM, relies on preliminary tuning and prior studies. We acknowledge that these settings may not be optimal. However, to mitigate bias, we maintained fixed configurations across all comparative methods. Future work could investigate automated hyperparameter optimization or sensitivity analysis to better understand their impact on \lisa's performance.
\paragraph{Final oracle verification and documentation coverage}
Although \lisa uses invariants and iterative repair to mitigate the oracle problem, final verification still requires human inspection. In particular, a failing assertion may indicate either a genuine implementation defect or a misinterpretation of the specification by the LLM. By design, \lisa's chunk-invariant reasoning and verification-and-repair loop produce high-confidence assertions; nevertheless, residual false positives cannot be fully eliminated. \rev{A second, orthogonal limit is recall: \lisa can only assert what documentation (plus filtered Daikon candidates) makes explicit, so conventional-but-undocumented contracts that experienced developers take for granted yield no invariant, and \lisa can miss bugs that violate only such unwritten conventions. Grounding invariants in documentation is a deliberate trade-off, keeping them semantically valid and low-noise at the cost of leaving undocumented intent to future work on mining usage patterns and existing test suites.} More broadly, the test oracle problem remains a fundamental open challenge in software testing~\cite{Barr2015oracle}, and resolving it more fully is an important direction for future work.
\paragraph{Automation}
Constructing the API knowledge database requires one-time manual effort, so the overall pipeline is not yet fully automated. In principle, this step could be delegated to an AI agent to further reduce manual effort. However, it remains unclear whether such automation can match the quality and reliability of human curation. Developing and evaluating agent-based alternatives for this stage is an important direction for future work.

\subsection{\rev{Future Work}}
\rev{Beyond automating knowledge-base curation and mining undocumented conventions, as discussed above, two extensions look promising. First, integrating richer API-dependency models, such as those used by Hopper~\cite{chen2023hopper} and \citywalk~\cite{zhang2025citywalk}, could further raise the validity of generated sequences beyond what the $n$-gram model captures. Second, extending \lisa to managed-language libraries (Python, Java) would test the portability of the two-stage decomposition beyond the C/C++ setting we studied here.}
\section{Related Work}

\subsection{Automated Software Testing for Libraries}

\paragraph{Library fuzzing}
Library fuzzing typically uses LLVM's \libfuzzer~\cite{serebryany2016libfuzzer}, for which developers write \emph{fuzz drivers} that parse inputs and invoke APIs; \ossfuzz~\cite{serebryany2017ossfuzz} maintains such drivers for over \num{1000} projects. As writing drivers requires expert knowledge of library semantics, recent work automates their generation~\cite{chen2023hopper,lyu2024promptfuzz,zhang2026llamafuzz}: Hopper~\cite{chen2023hopper} models API calls with a lightweight interpreter and grammar to explore valid API compositions (\emph{API sequences} in this paper), and \promptfuzz~\cite{lyu2024promptfuzz} prompts an LLM to generate fuzz drivers, iteratively mutating the API combinations under coverage feedback to reach deep implementation bugs.

\paragraph{Automated unit testing}
Heuristic-based unit test generation has been studied for decades but targets mainly object-oriented software (Java, C\#), generalizing poorly to C/C++ because it focuses on object internal state rather than library-level API interactions~\cite{pacheco2007randoop,fraser2011evosuite}. \randoop~\cite{pacheco2007randoop} mutates method-call sequences and checks language-level \emph{oracle} properties, while \mutest and \evosuite generate finer-grained oracles via mutation analysis; however, the injected artificial bugs are often unrealistic~\cite{panichella2020revisiting,rao2023catlm}, leaving these heuristics with low maintainability and usability~\cite{panichella2020revisiting}.

\subsection{LLM-Based Unit Test Generation}

To overcome the limitations of heuristic-based unit test generation,
recent studies have proposed using the code under test as prompts to LLMs to automatically generate unit tests~\cite{watson2020learning_assert,dinella2022TOGA,nie2023TeCo,rao2023catlm,he2024unitsyn,Ryan2024code,he2025fuzzaug,zhang2025citywalk}.
This line of work has evolved from generating only assertions~\cite{watson2020learning_assert,dinella2022TOGA,nie2023TeCo}
to producing complete test functions and even full test files~\cite{rao2023catlm,he2024unitsyn,he2025fuzzaug,zhang2025citywalk}.

Several benchmarks have been proposed to evaluate LLM-based software testing~\cite{wang2025testeval,mundler2024swtbench,jain2025testgeneval}.
TestEval~\cite{wang2025testeval} and SWT-Bench~\cite{mundler2024swtbench} focus on generating single test functions,
whereas TestGenEval~\cite{jain2025testgeneval} evaluates the generation of complete test suites, similar to \evosuite.
SWT-Bench measures bug detection capability by generating tests for bug-fixing pull requests (PRs),
where the generated tests should fail before the PR is merged and pass afterward.
TestGenEval incorporates mutation score to assess whether the generated tests can detect behavioral changes in the code.
However, no existing benchmark provides a ground-truth oracle for test generation,
which is necessary for evaluating whether generated tests can detect real-world functional bugs in the current versions of software libraries.
Moreover, all existing benchmarks support only Python,
while more widely used languages in the software supply chain, such as C and C++,
remain largely neglected.

\subsection{Invariant Generation}
Program invariant generation methods can be broadly categorized into
dynamic execution-based~\cite{ernst2001daikon},
search-based~\cite{si2020code2inv},
and learning-based approaches~\cite{si2020code2inv,pei2023can,wang2024smartinv}.
Execution-based methods, such as Daikon~\cite{ernst2001daikon},
capture invariants of the \emph{current} version of the software and are therefore better suited for generating regression tests rather than detecting functional bugs.
LLM-based approaches, such as SmartInv~\cite{wang2024smartinv},
infer invariants that \emph{should hold} according to the developer's intent
by reasoning over both source code semantics and corresponding natural language descriptions~\cite{he2025evaluating}.
Such invariants can be used to detect functional bugs in software libraries.
\lisa is a generic framework capable of integrating any style of invariant generation method based on user preference,
supporting both regression testing with Daikon and functional bug detection using LLM-based invariant generalization.
\rev{\lisa reconciles the two lines: Daikon supplies candidate relations that \lisa filters against documentation, while the LLM supplies intent-level contracts from the same documentation, checked in turn against the reference build. Neither source alone flags functional bugs, but their documentation-grounded reconciliation does.}

\subsection{\rev{Positioning of \lisa}}
\rev{\lisa shares individual building blocks with prior work, but its contribution lies in how it \emph{decomposes} functional-bug detection rather than in any single component. \promptfuzz~\cite{lyu2024promptfuzz} generates \emph{fuzz drivers} for \emph{crash} detection and steers exploration by mutating API combinations under code-coverage feedback; \lisa instead synthesizes \emph{API sequences} (not drivers), targets \emph{functional} bugs, steers exploration with a new $n$-gram API-combination coverage (\autoref{sec:n-gram}), and supplies the invariant oracle that \promptfuzz lacks, reusing only \promptfuzz's front-end API extraction, not its fuzzing loop. SmartInv~\cite{wang2024smartinv} infers invariants for \emph{smart contracts} from code and intent; \lisa instead targets \emph{C/C++ library APIs}, derives invariants from API \emph{documentation} as executable test oracles, and, unlike monolithic inference, \emph{decouples} oracle construction from sequence generation, inserting invariants incrementally at chunk boundaries under a verify-and-repair loop (\autoref{sec:chunk-inv-gen}). LLM-UT methods such as \citywalk~\cite{zhang2025citywalk} follow a one-test-per-method scheme whose oracles are frequently trivial (e.g., \texttt{assert(ptr != NULL)}) and cannot construct the multi-step API sequences required to reach stateful functional bugs. The novelty of \lisa is therefore not any single technique it borrows, but this decomposition, separating reachability (sequence synthesis) from oracle construction, guiding the former with $n$-gram API-combination coverage and grounding the latter in documentation-derived chunk invariants, which no prior approach provides.}
\section{Conclusion}

We presented \lisa, a novel LLM-based unit testing framework for detecting functional bugs in software libraries.
By decoupling \apiseq exploration from invariant construction in unit tests,
\lisa leverages \rev{$n$-gram API-combination} coverage feedback and chunk-invariant reasoning to generate diverse and semantically meaningful test functions.
\rev{Across \nbench C/C++ libraries, it attains higher average branch coverage than \ossfuzz, and on \num{25} historical functional bugs it detects \num{12} (48\%), versus \num{3} for \citywalk and \num{2} for \ossfuzz. Its novelty lies in this decomposition rather than in any single borrowed component, and it mitigates rather than solves the oracle problem, emitting high-confidence bug candidates for developer confirmation.}

\section*{Data Availability}

Our artifacts and source code to reproduce the data are available at \cite{lisa_zenodo} and on GitHub.\footnote{\url{https://github.com/SecurityLab-UCD/CNTG} and \url{https://github.com/SecurityLab-UCD/CGNTG}} This public digital repository includes: (1) the \lisa implementation and prompts, (2) the full list of the \rev{25} re-introduced historical bugs (IDs, links, and corresponding patches/commits), (3) the generated test programs for the bugs found by \lisa, (4) all scripts/configurations used for coverage measurement and reporting, and (5) the rationale and sensitivity data for choosing the $n$-gram parameter $N$.

%


\printbibliography

@inproceedings{hazimeh2020magma,
  author     = {Hazimeh, Ahmad and Herrera, Adrian and Payer, Mathias},
  title      = {Magma: A Ground-Truth Fuzzing Benchmark},
  year       = {2020},
  issue_date = {December 2020},
  publisher  = {Association for Computing Machinery},
  address    = {New York, NY, USA},
  volume     = {4},
  number     = {3},
  doi        = {10.1145/3428334},
  journal    = {Proc. ACM Meas. Anal. Comput. Syst.},
  month      = nov,
  articleno  = {49},
  numpages   = {29}
}

@inproceedings{metzman2021fuzzbench,
  author    = {Metzman, Jonathan and Szekeres, L\'{a}szl\'{o} and Simon, Laurent and Sprabery, Read and Arya, Abhishek},
  title     = {FuzzBench: an open fuzzer benchmarking platform and service},
  year      = {2021},
  isbn      = {9781450385626},
  publisher = {Association for Computing Machinery},
  address   = {New York, NY, USA},
  doi       = {10.1145/3468264.3473932},
  booktitle = {Proceedings of the 29th ACM Joint Meeting on European Software Engineering Conference and Symposium on the Foundations of Software Engineering},
  pages     = {1393–1403},
  numpages  = {11},
  location  = {Athens, Greece},
  series    = {ESEC/FSE 2021}
}

@inproceedings{he2024unitsyn,
  author    = {He, Yifeng and Huang, Jiabo and Rong, Yuyang and Guo, Yiwen and Wang, Ethan and Chen, Hao},
  title     = {{UniTSyn}: A Large-Scale Dataset Capable of Enhancing the Prowess of Large Language Models for Program Testing},
  year      = {2024},
  isbn      = {9798400706127},
  publisher = {Association for Computing Machinery},
  address   = {New York, NY, USA},
  doi       = {10.1145/3650212.3680342},
  booktitle = {Proceedings of the 33rd ACM SIGSOFT International Symposium on Software Testing and Analysis},
  pages     = {1061–1072},
  numpages  = {12},
  location  = {Vienna, Austria},
  series    = {ISSTA 2024}
}

@inproceedings{he2025fuzzaug,
  title     = {{F}uzz{A}ug: Data Augmentation by Coverage-guided Fuzzing for Neural Test Generation},
  author    = {He, Yifeng  and
               Wang, Jicheng  and
               Rong, Yuyang  and
               Chen, Hao},
  booktitle = {Findings of the Association for Computational Linguistics: EMNLP 2025},
  month     = nov,
  year      = {2025},
  address   = {Suzhou, China},
  publisher = {Association for Computational Linguistics},
  doi       = {10.18653/v1/2025.findings-emnlp.847},
  pages     = {15642--15655},
  isbn      = {979-8-89176-335-7}
}

@inproceedings{he2025evaluating,
  title     = {Evaluating Program Semantics Reasoning with Type Inference in System \$F\$},
  author    = {Yifeng He and Luning Yang and Christopher Castro Gaw Gonzalo and Hao Chen},
  booktitle = {The Thirty-ninth Annual Conference on Neural Information Processing Systems Datasets and Benchmarks Track},
  year      = {2025},
  url       = {https://openreview.net/forum?id=IA9RmaP0aw}
}

@inproceedings{chen2023hopper,
  author    = {Chen, Peng and Xie, Yuxuan and Lyu, Yunlong and Wang, Yuxiao and Chen, Hao},
  title     = {Hopper: Interpretative Fuzzing for Libraries},
  year      = {2023},
  isbn      = {9798400700507},
  publisher = {Association for Computing Machinery},
  address   = {New York, NY, USA},
  doi       = {10.1145/3576915.3616610},
  booktitle = {Proceedings of the 2023 ACM SIGSAC Conference on Computer and Communications Security},
  pages     = {1600–1614},
  numpages  = {15},
  location  = {Copenhagen, Denmark},
  series    = {CCS '23}
}

@inproceedings{lyu2024promptfuzz,
  author    = {Lyu, Yunlong and Xie, Yuxuan and Chen, Peng and Chen, Hao},
  title     = {Prompt Fuzzing for Fuzz Driver Generation},
  year      = {2024},
  isbn      = {9798400706363},
  publisher = {Association for Computing Machinery},
  address   = {New York, NY, USA},
  doi       = {10.1145/3658644.3670396},
  booktitle = {Proceedings of the 2024 on ACM SIGSAC Conference on Computer and Communications Security},
  pages     = {3793–3807},
  numpages  = {15},
  location  = {Salt Lake City, UT, USA},
  series    = {CCS '24}
}

@inproceedings{rao2023catlm,
  author    = {Rao, Nikitha and Jain, Kush and Alon, Uri and Goues, Claire Le and Hellendoorn, Vincent J.},
  title     = {CAT-LM Training Language Models on Aligned Code and Tests},
  year      = {2024},
  isbn      = {9798350329964},
  publisher = {IEEE Press},
  doi       = {10.1109/ASE56229.2023.00193},
  booktitle = {Proceedings of the 38th IEEE/ACM International Conference on Automated Software Engineering},
  pages     = {409–420},
  numpages  = {12},
  location  = {Echternach, Luxembourg},
  series    = {ASE '23}
}

@inproceedings{fraser2011evosuite,
  author    = {Fraser, Gordon and Arcuri, Andrea},
  title     = {EvoSuite: automatic test suite generation for object-oriented software},
  year      = {2011},
  isbn      = {9781450304436},
  publisher = {Association for Computing Machinery},
  address   = {New York, NY, USA},
  doi       = {10.1145/2025113.2025179},
  booktitle = {Proceedings of the 19th ACM SIGSOFT Symposium and the 13th European Conference on Foundations of Software Engineering},
  pages     = {416–419},
  numpages  = {4},
  location  = {Szeged, Hungary},
  series    = {ESEC/FSE '11}
}

@inproceedings{ernst2000daikon,
  author    = {Ernst, Michael D. and Czeisler, Adam and Griswold, William G. and Notkin, David},
  title     = {Quickly detecting relevant program invariants},
  year      = {2000},
  isbn      = {1581132069},
  publisher = {Association for Computing Machinery},
  address   = {New York, NY, USA},
  doi       = {10.1145/337180.337240},
  booktitle = {Proceedings of the 22nd International Conference on Software Engineering},
  pages     = {449–458},
  numpages  = {10},
  location  = {Limerick, Ireland},
  series    = {ICSE '00}
}

@article{ernst2001daikon,
  author   = {Ernst, M.D. and Cockrell, J. and Griswold, W.G. and Notkin, D.},
  journal  = {IEEE Transactions on Software Engineering},
  title    = {Dynamically discovering likely program invariants to support program evolution},
  year     = {2001},
  volume   = {27},
  number   = {2},
  pages    = {99-123},
  doi      = {10.1109/32.908957}
}

@inproceedings{pacheco2007randoop,
  author    = {Pacheco, Carlos and Lahiri, Shuvendu K. and Ernst, Michael D. and Ball, Thomas},
  booktitle = {29th International Conference on Software Engineering (ICSE'07)},
  title     = {Feedback-Directed Random Test Generation},
  year      = {2007},
  volume    = {},
  number    = {},
  pages     = {75-84},
  doi       = {10.1109/ICSE.2007.37}
}

@inproceedings{pei2023can,
  title     = {Can Large Language Models Reason about Program Invariants?},
  author    = {Pei, Kexin and Bieber, David and Shi, Kensen and Sutton, Charles and Yin, Pengcheng},
  booktitle = {Proceedings of the 40th International Conference on Machine Learning},
  pages     = {27496--27520},
  year      = {2023},
  volume    = {202},
  series    = {Proceedings of Machine Learning Research},
  month     = {23--29 Jul},
  publisher = {PMLR},
  url       = {https://proceedings.mlr.press/v202/pei23a.html}
}

@inproceedings{wang2024smartinv,
  author    = { Wang, Sally Junsong and Pei, Kexin and Yang, Junfeng },
  booktitle = { 2024 IEEE Symposium on Security and Privacy (SP) },
  title     = {{ SmartInv}: Multimodal Learning for Smart Contract Invariant Inference },
  year      = {2024},
  volume    = {},
  issn      = {},
  pages     = {2217-2235},
  doi       = {10.1109/SP54263.2024.00126},
  publisher = {IEEE Computer Society},
  address   = {Los Alamitos, CA, USA},
  month     = May
}

@article{zhang2025citywalk,
  author    = {Zhang, Yuwei and Lu, Qingyuan and Liu, Kai and Dou, Wensheng and Zhu, Jiaxin and Qian, Li and Zhang, Chunxi and Lin, Zheng and Wei, Jun},
  title     = {CITYWALK: Enhancing LLM-Based C++ Unit Test Generation via Project-Dependency Awareness and Language-Specific Knowledge},
  year      = {2025},
  publisher = {Association for Computing Machinery},
  address   = {New York, NY, USA},
  issn      = {1049-331X},
  doi       = {10.1145/3763791},
  note      = {Just Accepted},
  journal   = {ACM Trans. Softw. Eng. Methodol.},
  month     = aug
}

@conference{serebryany2017ossfuzz,
  author    = {Kostya Serebryany},
  title     = {{OSS-Fuzz} - Google{\textquoteright}s continuous fuzzing service for open source software},
  year      = {2017},
  address   = {Vancouver, BC},
  publisher = {USENIX Association},
  month     = aug
}

@misc{zlib,
  title        = {ZLIB DATA COMPRESSION LIBRARY},
  author       = {Jean-loup Gailly and Mark Adler},
  year         = {2025},
  howpublished = {\url{https://zlib.net/}}
}

@misc{cJSON,
  title        = {cJSON: Ultralightweight JSON Parser in ANSI C},
  author       = {Dave Gamble},
  year         = {2025},
  howpublished = {\url{https://github.com/DaveGamble/cJSON}}
}

@misc{re2,
  title        = {RE2: Efficient Regular Expression Matching},
  author       = {{Google team}},
  year         = {2025},
  howpublished = {\url{https://github.com/google/re2}}
}

@misc{libpng,
  title        = {LIBPNG: Portable Network Graphics Reference Library},
  author       = {{PNG Development Group}},
  year         = {2025},
  howpublished = {\url{https://github.com/pnggroup/libpng}}
}

@misc{libpcap,
  title        = {libpcap: Portable Packet Capture Library},
  author       = {{The Tcpdump Group}},
  year         = {2025},
  howpublished = {\url{https://github.com/the-tcpdump-group/libpcap}}
}

@misc{sqlite3,
  title        = {SQLite},
  author       = {{SQLite Development Team}},
  year         = {2025},
  howpublished = {\url{https://github.com/sqlite/sqlite}}
}

@misc{lcms,
  title        = {Little CMS: A Free Color Management Engine},
  author       = {Marti Maria},
  year         = {2025},
  howpublished = {\url{https://github.com/mm2/Little-CMS}}
}

@inproceedings{jimenez2024swebench,
  author    = {Jimenez, Carlos E. and Yang, John and Wettig, Alexander and Yao, Shunyu and Pei, Kexin and Press, Ofir and Narasimhan, Karthik},
  title     = {{SWE}-bench: Can Language Models Resolve Real-World {GitHub} Issues?},
  booktitle = {The Twelfth International Conference on Learning Representations (ICLR)},
  year      = {2024}
}

@inproceedings{yang2024sweagent,
  author    = {Yang, John and Jimenez, Carlos E. and Wettig, Alexander and Lieret, Kilian and Yao, Shunyu and Narasimhan, Karthik and Press, Ofir},
  title     = {{SWE}-agent: Agent-Computer Interfaces Enable Automated Software Engineering},
  booktitle = {Advances in Neural Information Processing Systems (NeurIPS)},
  year      = {2024}
}

@inproceedings{wang2025openhands,
  author    = {Wang, Xingyao and others},
  title     = {{OpenHands}: An Open Platform for {AI} Software Developers as Generalist Agents},
  booktitle = {The Thirteenth International Conference on Learning Representations (ICLR)},
  year      = {2025}
}

@inproceedings{nan2025intut,
  author    = {Nan, Zifan and Guo, Zhaoqiang and Liu, Kui and Xia, Xin},
  title     = {Test Intention Guided {LLM}-based Unit Test Generation},
  booktitle = {Proceedings of the IEEE/ACM 47th International Conference on Software Engineering (ICSE)},
  year      = {2025},
  doi       = {10.1109/ICSE55347.2025.00243}
}

@misc{liu2025llmbasedunittestgeneration,
  title         = {LLM-based Unit Test Generation for Dynamically-Typed Programs},
  author        = {Runlin Liu and Zhe Zhang and Yunge Hu and Yuhang Lin and Xiang Gao and Hailong Sun},
  year          = {2025},
  eprint        = {2503.14000},
  archiveprefix = {arXiv},
  primaryclass  = {cs.SE},
  url           = {https://arxiv.org/abs/2503.14000}
}

@article{fenton2000quantitative,
  author   = {Fenton, N.E. and Ohlsson, N.},
  journal  = {IEEE Transactions on Software Engineering},
  title    = {Quantitative analysis of faults and failures in a complex software system},
  year     = {2000},
  volume   = {26},
  number   = {8},
  pages    = {797-814},
  doi      = {10.1109/32.879815}
}

@inproceedings{zhao2017ci,
  author    = {Zhao, Yangyang and Serebrenik, Alexander and Zhou, Yuming and Filkov, Vladimir and Vasilescu, Bogdan},
  booktitle = {2017 32nd IEEE/ACM International Conference on Automated Software Engineering (ASE)},
  title     = {The impact of continuous integration on other software development practices: A large-scale empirical study},
  year      = {2017},
  url       = {https://doi.org/10.1109/ASE.2017.8115619}
}

@inproceedings{beller2015when,
  author    = {Beller, Moritz and Gousios, Georgios and Panichella, Annibale and Zaidman, Andy},
  title     = {When, How, and Why Developers (Do Not) Test in Their IDEs},
  year      = {2015},
  isbn      = {9781450336758},
  publisher = {Association for Computing Machinery},
  address   = {New York, NY, USA},
  url       = {https://doi.org/10.1145/2786805.2786843},
  booktitle = {Proceedings of the 2015 10th Joint Meeting on Foundations of Software Engineering},
  pages     = {179–190},
  numpages  = {12},
  location  = {Bergamo, Italy},
  series    = {ESEC/FSE 2015}
}

@inproceedings{serebryany2016libfuzzer,
  author    = {Serebryany, Kosta},
  booktitle = {2016 IEEE Cybersecurity Development (SecDev)},
  title     = {Continuous Fuzzing with libFuzzer and AddressSanitizer},
  year      = {2016},
  volume    = {},
  number    = {},
  pages     = {157-157},
  url       = {https://doi.org/10.1109/SecDev.2016.043}
}

@inproceedings{fioraldi2020afl_pp,
  title     = {AFL++ combining incremental steps of fuzzing research},
  author    = {Fioraldi, Andrea and Maier, Dominik and Ei{\ss}feldt, Heiko and Heuse, Marc},
  booktitle = {Proceedings of the 14th USENIX Conference on Offensive Technologies},
  pages     = {10--10},
  year      = {2020}
}

@article{chou2001os,
  author     = {Chou, Andy and Yang, Junfeng and Chelf, Benjamin and Hallem, Seth and Engler, Dawson},
  title      = {An empirical study of operating systems errors},
  year       = {2001},
  issue_date = {Dec. 2001},
  publisher  = {Association for Computing Machinery},
  address    = {New York, NY, USA},
  volume     = {35},
  number     = {5},
  issn       = {0163-5980},
  doi        = {10.1145/502059.502042},
  journal    = {SIGOPS Oper. Syst. Rev.},
  month      = oct,
  pages      = {73–88},
  numpages   = {16}
}

@misc{zhu2025bugsbenchmarks,
  title         = {From Bugs to Benchmarks: A Comprehensive Survey of Software Defect Datasets},
  author        = {Hao-Nan Zhu and Robert M. Furth and Michael Pradel and Cindy Rubio-González},
  year          = {2025},
  eprint        = {2504.17977},
  archiveprefix = {arXiv},
  primaryclass  = {cs.SE},
  url           = {https://arxiv.org/abs/2504.17977}
}

@inproceedings{bohme2017directed,
  author    = {B\"{o}hme, Marcel and Pham, Van-Thuan and Nguyen, Manh-Dung and Roychoudhury, Abhik},
  title     = {Directed Greybox Fuzzing},
  year      = {2017},
  isbn      = {9781450349468},
  publisher = {Association for Computing Machinery},
  address   = {New York, NY, USA},
  doi       = {10.1145/3133956.3134020},
  booktitle = {Proceedings of the 2017 ACM SIGSAC Conference on Computer and Communications Security},
  pages     = {2329–2344},
  numpages  = {16},
  location  = {Dallas, Texas, USA},
  series    = {CCS '17}
}

@book{khorikov2020unit,
  title     = {Unit Testing Principles, Practices, and Patterns},
  author    = {Khorikov, Vladimir},
  year      = {2020},
  publisher = {Simon and Schuster}
}

@inproceedings{jain2025testgeneval,
  title     = {TestGenEval: A Real World Unit Test Generation and Test Completion Benchmark},
  author    = {Kush Jain and Gabriel Synnaeve and Baptiste Roziere},
  booktitle = {The Thirteenth International Conference on Learning Representations},
  year      = {2025},
  url       = {https://openreview.net/forum?id=7o6SG5gVev}
}

@article{schafer2024An,
  author   = {Schäfer, Max and Nadi, Sarah and Eghbali, Aryaz and Tip, Frank},
  journal  = {IEEE Transactions on Software Engineering},
  title    = {An Empirical Evaluation of Using Large Language Models for Automated Unit Test Generation},
  year     = {2024},
  volume   = {50},
  number   = {1},
  pages    = {85-105},
  doi      = {10.1109/TSE.2023.3334955}
}

@article{SANTOS2017Stepwise,
  title    = {Stepwise API usage assistance using n-gram language models},
  author   = {Andr{\'e} L. Santos and Gon{\c c}alo Prendi and Hugo Sousa and Ricardo Ribeiro},
  journal  = {Journal of Systems and Software},
  volume   = {131},
  pages    = {461--474},
  year     = {2017},
  issn     = {0164-1212},
  doi      = {10.1016/j.jss.2016.06.063}
}

@inproceedings{nie2023TeCo,
  author    = {Nie, Pengyu and Banerjee, Rahul and Li, Junyi Jessy and Mooney, Raymond J. and Gligoric, Milos},
  title     = {Learning Deep Semantics for Test Completion},
  year      = {2023},
  isbn      = {9781665457019},
  publisher = {IEEE Press},
  doi       = {10.1109/ICSE48619.2023.00178},
  booktitle = {Proceedings of the 45th International Conference on Software Engineering},
  pages     = {2111–2123},
  numpages  = {13},
  location  = {Melbourne, Victoria, Australia},
  series    = {ICSE '23}
}

@inproceedings{panichella2020revisiting,
  title        = {Revisiting test smells in automatically generated tests: limitations, pitfalls, and opportunities},
  author       = {Panichella, Annibale and Panichella, Sebastiano and Fraser, Gordon and Sawant, Anand Ashok and Hellendoorn, Vincent J},
  booktitle    = {2020 IEEE international conference on software maintenance and evolution (ICSME)},
  pages        = {523--533},
  year         = {2020},
  organization = {IEEE},
  doi          = {10.1109/ICSME46990.2020.00056}
}

@misc{llvm-project,
  title = {The LLVM Compiler Infrastructure},
  url   = {https://github.com/llvm/llvm-project},
  note  = {\url{https://github.com/llvm/llvm-project}, accessed 2025-10-12}
}

@article{Ryan2024code,
  author     = {Ryan, Gabriel and Jain, Siddhartha and Shang, Mingyue and Wang, Shiqi and Ma, Xiaofei and Ramanathan, Murali Krishna and Ray, Baishakhi},
  title      = {Code-Aware Prompting: A Study of Coverage-Guided Test Generation in Regression Setting using LLM},
  year       = {2024},
  issue_date = {July 2024},
  publisher  = {Association for Computing Machinery},
  address    = {New York, NY, USA},
  volume     = {1},
  number     = {FSE},
  doi        = {10.1145/3643769},
  journal    = {Proc. ACM Softw. Eng.},
  month      = jul,
  articleno  = {43},
  numpages   = {21}
}

@inproceedings{wan-etal-2023-better,
  title     = {Better Zero-Shot Reasoning with Self-Adaptive Prompting},
  author    = {Wan, Xingchen  and
               Sun, Ruoxi  and
               Dai, Hanjun  and
               Arik, Sercan  and
               Pfister, Tomas},
  booktitle = {Findings of the Association for Computational Linguistics: ACL 2023},
  month     = jul,
  year      = {2023},
  address   = {Toronto, Canada},
  publisher = {Association for Computational Linguistics},
  doi       = {10.18653/v1/2023.findings-acl.216},
  pages     = {3493--3514}
}

@article{Barr2015oracle,
  author   = {Barr, Earl T. and Harman, Mark and McMinn, Phil and Shahbaz, Muzammil and Yoo, Shin},
  journal  = {IEEE Transactions on Software Engineering},
  title    = {The Oracle Problem in Software Testing: A Survey},
  year     = {2015},
  volume   = {41},
  number   = {5},
  pages    = {507-525},
  doi      = {10.1109/TSE.2014.2372785}
}

@article{liu-etal-2024-lost,
  title     = {Lost in the Middle: How Language Models Use Long Contexts},
  author    = {Liu, Nelson F.  and
               Lin, Kevin  and
               Hewitt, John  and
               Paranjape, Ashwin  and
               Bevilacqua, Michele  and
               Petroni, Fabio  and
               Liang, Percy},
  journal   = {Transactions of the Association for Computational Linguistics},
  volume    = {12},
  year      = {2024},
  address   = {Cambridge, MA},
  publisher = {MIT Press},
  doi       = {10.1162/tacl_a_00638},
  pages     = {157--173}
}

@inproceedings{shamshiri2018how,
  author    = {Shamshiri, Sina and Rojas, José Miguel and Galeotti, Juan Pablo and Walkinshaw, Neil and Fraser, Gordon},
  booktitle = {2018 IEEE 11th International Conference on Software Testing, Verification and Validation (ICST)},
  title     = {How Do Automatically Generated Unit Tests Influence Software Maintenance?},
  year      = {2018},
  volume    = {},
  number    = {},
  pages     = {250-261},
  doi       = {10.1109/ICST.2018.00033}
}

@inproceedings{watson2020learning_assert,
  author    = {Watson, Cody and Tufano, Michele and Moran, Kevin and Bavota, Gabriele and Poshyvanyk, Denys},
  title     = {On learning meaningful assert statements for unit test cases},
  year      = {2020},
  isbn      = {9781450371216},
  publisher = {Association for Computing Machinery},
  address   = {New York, NY, USA},
  doi       = {10.1145/3377811.3380429},
  booktitle = {Proceedings of the ACM/IEEE 42nd International Conference on Software Engineering},
  pages     = {1398–1409},
  numpages  = {12},
  location  = {Seoul, South Korea},
  series    = {ICSE '20}
}

@inproceedings{dinella2022TOGA,
  author    = {Dinella, Elizabeth and Ryan, Gabriel and Mytkowicz, Todd and Lahiri, Shuvendu K.},
  title     = {TOGA: A Neural Method for Test Oracle Generation},
  year      = {2022},
  isbn      = {9781450392211},
  publisher = {Association for Computing Machinery},
  address   = {New York, NY, USA},
  url       = {https://doi.org/10.1145/3510003.3510141},
  booktitle = {Proceedings of the 44th International Conference on Software Engineering},
  pages     = {2130–2141},
  numpages  = {12},
  location  = {Pittsburgh, Pennsylvania},
  series    = {ICSE '22}
}

@inproceedings{wang2025testeval,
  title     = {{TestEval}: Benchmarking Large Language Models for Test Case Generation},
  author    = {Wang, Wenhan  and
               Yang, Chenyuan  and
               Wang, Zhijie  and
               Huang, Yuheng  and
               Chu, Zhaoyang  and
               Song, Da  and
               Zhang, Lingming  and
               Chen, An Ran  and
               Ma, Lei},
  booktitle = {Findings of the Association for Computational Linguistics: NAACL 2025},
  month     = apr,
  year      = {2025},
  address   = {Albuquerque, New Mexico},
  publisher = {Association for Computational Linguistics},
  doi       = {10.18653/v1/2025.findings-naacl.197},
  pages     = {3547--3562},
  isbn      = {979-8-89176-195-7}
}

@inproceedings{mundler2024swtbench,
  author    = {M\"{u}ndler, Niels and M\"{u}ller, Mark Niklas and He, Jingxuan and Vechev, Martin},
  booktitle = {Advances in Neural Information Processing Systems},
  pages     = {81857--81887},
  publisher = {Curran Associates, Inc.},
  title     = {SWT-Bench: Testing and Validating Real-World Bug-Fixes with Code Agents},
  url       = {https://proceedings.neurips.cc/paper_files/paper/2024/file/94f093b41fc2666376fb1f667fe282f3-Paper-Conference.pdf},
  volume    = {37},
  year      = {2024}
}

@inproceedings{Papadakis2018Are,
  author    = {Papadakis, Mike and Shin, Donghwan and Yoo, Shin and Bae, Doo-Hwan},
  booktitle = {2018 IEEE/ACM 40th International Conference on Software Engineering (ICSE)},
  title     = {Are Mutation Scores Correlated with Real Fault Detection? A Large Scale Empirical Study on the Relationship Between Mutants and Real Faults},
  year      = {2018},
  volume    = {},
  number    = {},
  pages     = {537-548},
  doi       = {10.1145/3180155.3180183}
}

@inproceedings{Andrews2005Is,
  author    = {Andrews, J. H. and Briand, L. C. and Labiche, Y.},
  title     = {Is mutation an appropriate tool for testing experiments?},
  year      = {2005},
  isbn      = {1581139632},
  publisher = {Association for Computing Machinery},
  address   = {New York, NY, USA},
  doi       = {10.1145/1062455.1062530},
  booktitle = {Proceedings of the 27th International Conference on Software Engineering},
  pages     = {402–411},
  numpages  = {10},
  location  = {St. Louis, MO, USA},
  series    = {ICSE '05}
}

@inproceedings{si2020code2inv,
  author    = {Si, Xujie and Naik, Aaditya and Dai, Hanjun and Naik, Mayur and Song, Le},
  title     = {Code2Inv: A Deep Learning Framework for Program Verification},
  year      = {2020},
  isbn      = {978-3-030-53290-1},
  publisher = {Springer-Verlag},
  address   = {Berlin, Heidelberg},
  doi       = {10.1007/978-3-030-53291-8_9},
  booktitle = {Computer Aided Verification: 32nd International Conference, CAV 2020, Los Angeles, CA, USA, July 21–24, 2020, Proceedings, Part II},
  pages     = {151–164},
  numpages  = {14},
  location  = {Los Angeles, CA, USA}
}

@inproceedings{Xie2022DocTer,
  author    = {Xie, Danning and Li, Yitong and Kim, Mijung and Pham, Hung Viet and Tan, Lin and Zhang, Xiangyu and Godfrey, Michael W.},
  title     = {DocTer: documentation-guided fuzzing for testing deep learning API functions},
  year      = {2022},
  isbn      = {9781450393799},
  publisher = {Association for Computing Machinery},
  address   = {New York, NY, USA},
  doi       = {10.1145/3533767.3534220},
  booktitle = {Proceedings of the 31st ACM SIGSOFT International Symposium on Software Testing and Analysis},
  pages     = {176–188},
  numpages  = {13},
  location  = {Virtual, South Korea},
  series    = {ISSTA 2022}
}

@misc{lisa_zenodo,
  author       = {Yang, Ruogu and He, Yifeng and Xu, Yundi and Wei, Yuqing and Chen, Hao},
  title        = {Lisa Artifacts},
  month        = {1},
  year         = {2026},
  doi          = {10.5281/zenodo.20839021},
  howpublished = {Zenodo, \url{https://doi.org/10.5281/zenodo.20839021}},
  publisher    = {Zenodo}
}

@inproceedings{claessen2000quick_check,
  title     = {{QuickCheck}: A Lightweight Tool for Random Testing of {Haskell} Programs},
  author    = {Claessen, Koen and Hughes, John},
  booktitle = {5th ACM SIGPLAN International Conference on Functional Programming (ICFP)},
  pages     = {268--279},
  year      = {2000},
  publisher = {ACM}
}

@article{lampropoulos2019coverage,
  title     = {Coverage Guided, Property Based Testing},
  author    = {Lampropoulos, Leonidas and Hicks, Michael and Pierce, Benjamin C.},
  journal   = {Proceedings of the ACM on Programming Languages (OOPSLA)},
  volume    = {3},
  pages     = {1--29},
  year      = {2019},
  publisher = {ACM}
}

@incollection{dijkstra1972notes,
  title     = {Notes on Structured Programming},
  author    = {Dijkstra, Edsger W.},
  booktitle = {Structured Programming},
  pages     = {1--82},
  year      = {1972},
  publisher = {Academic Press}
}

@inproceedings{lattner2004llvm,
  title     = {{LLVM}: A Compilation Framework for Lifelong Program Analysis \& Transformation},
  author    = {Lattner, Chris and Adve, Vikram},
  booktitle = {International Symposium on Code Generation and Optimization (CGO)},
  pages     = {75--86},
  year      = {2004},
  publisher = {IEEE}
}

@article{bessey2010few,
  title     = {A Few Billion Lines of Code Later: Using Static Analysis to Find Bugs in the Real World},
  author    = {Bessey, Al and Block, Ken and Chelf, Ben and Chou, Andy and Fulton, Bryan and Hallem, Seth and Henri-Gros, Charles and Kamsky, Asya and McPeak, Scott and Engler, Dawson},
  journal   = {Communications of the ACM},
  volume    = {53},
  number    = {2},
  pages     = {66--75},
  year      = {2010},
  publisher = {ACM}
}

@inproceedings{cadar2008klee,
  title     = {{KLEE}: Unassisted and Automatic Generation of High-Coverage Tests for Complex Systems Programs},
  author    = {Cadar, Cristian and Dunbar, Daniel and Engler, Dawson R.},
  booktitle = {8th USENIX Symposium on Operating Systems Design and Implementation (OSDI)},
  pages     = {209--224},
  year      = {2008}
}

@inproceedings{leino2010dafny,
  title     = {Dafny: An Automatic Program Verifier for Functional Correctness},
  author    = {Leino, K. Rustan M.},
  booktitle = {16th International Conference on Logic for Programming, Artificial Intelligence, and Reasoning (LPAR)},
  pages     = {348--370},
  year      = {2010},
  publisher = {Springer}
}

@misc{wang2024repogenreflex,
  title         = {RepoGenReflex: Enhancing Repository-Level Code Completion with Verbal Reinforcement and Retrieval-Augmented Generation},
  author        = {Jicheng Wang and Yifeng He and Hao Chen},
  year          = {2024},
  eprint        = {2409.13122},
  archiveprefix = {arXiv},
  primaryclass  = {cs.SE},
  url           = {https://arxiv.org/abs/2409.13122}
}

@inproceedings{shinn2023reflexion,
  title     = {Reflexion: language agents with verbal reinforcement learning},
  author    = {Noah Shinn and Federico Cassano and Ashwin Gopinath and Karthik R Narasimhan and Shunyu Yao},
  booktitle = {Thirty-seventh Conference on Neural Information Processing Systems},
  year      = {2023},
  url       = {https://openreview.net/forum?id=vAElhFcKW6}
}

@inproceedings{zhang2026llamafuzz,
  title     = {LLAMAFUZZ: Large Language Model Enhanced Greybox Fuzzing},
  author    = {Hongxiang Zhang and Yuyang Rong and Yifeng He and Hao Chen},
  booktitle = {ACM/IEEE International Conference on Automation of Software Test (AST)},
  year      = {2026},
  url       = {https://arxiv.org/abs/2406.07714}
}

@article{xia2024agentless,
  author     = {Xia, Chunqiu Steven and Deng, Yinlin and Dunn, Soren and Zhang, Lingming},
  title      = {Demystifying LLM-Based Software Engineering Agents},
  year       = {2025},
  issue_date = {July 2025},
  publisher  = {Association for Computing Machinery},
  address    = {New York, NY, USA},
  volume     = {2},
  number     = {FSE},
  doi        = {10.1145/3715754},
  journal    = {Proc. ACM Softw. Eng.},
  month      = jun,
  articleno  = {FSE037},
  numpages   = {24}
}

@article{weyuker1982nontestable,
  author  = {Weyuker, Elaine J.},
  title   = {On Testing Non-testable Programs},
  journal = {The Computer Journal},
  volume  = {25},
  number  = {4},
  pages   = {465--470},
  year    = {1982},
  doi     = {10.1093/comjnl/25.4.465}
}

@inproceedings{howden1978theoretical,
  author    = {Howden, William E.},
  title     = {Theoretical and Empirical Studies of Program Testing},
  booktitle = {Proceedings of the 3rd International Conference on Software Engineering (ICSE)},
  series    = {ICSE '78},
  pages     = {305--311},
  year      = {1978},
  publisher = {IEEE Press},
  location  = {Atlanta, Georgia, USA}
}

\clearpage 
\newpage
\appendices
\section{Prompt Templates} \label{sec:prompt-templates}

\begin{figure}[htbp]
	\input{figures/prompt1.tex}
	\caption{The structured prompt template employed in our framework. Dynamic fields are denoted by brackets.}
	\label{fig:prompt_template}
\end{figure}

\begin{figure}[htbp]
	\input{figures/repair.tex}
	\caption{The structured repair prompt template for correcting erroneous API sequences. Dynamic fields are denoted by brackets.}
	\label{fig:repair}
\end{figure}

\section{Details of API N-Gram Feedback}

\subsection{Proofs of Condense Transformation} \label{sec:proofs-of-condense-transformation}
\begin{proof}[Proof of \autoref{prop:range-compression}]
	Let $x, y \in I$ such that $0 < x < y \leq 1$.
	Let $\lambda = x / y \in (0, 1)$,
	then $x = \lambda y + (1 - \lambda) \cdot 0$.
	By strict concavity,
	\[
		\lambda f(y) = \lambda f(y) + (1 - \lambda) f(0) < f(\lambda y + (1 - \lambda) \cdot 0) = f(x).
	\]
	Rearranging the above inequality gives
	\[
		\lambda f(y) < f(x) \implies f(y) < \frac{f(x)}{\lambda} = \frac{y}{x} f(x) \implies \frac{f(y)}{f(x)} < \frac{y}{x}.
	\]
	Since $f$ is strictly monotone increasing, we have $0 < f(x) < f(y)$,
	thus
	\[
		1 < \frac{f(y)}{f(x)} < \frac{y}{x}.
	\]
	We conclude that any strictly monotone increasing concave function with left-endpoint preserved compresses the range of input values.
\end{proof}

\begin{proof}[Proof that power condense is a condense transformation]
	We prove that power condense is a condense transformation.
	We write the power condense $C(a) = \hat{E}(a)^{\alpha}$ as $f(x) = x^{\alpha}$ for $\alpha \in (0, 1)$ for simplification.
	\begin{itemize}
		\item \textit{Strictly monotone increasing:}
		      We take the first derivative of $f$ with respect to $x$,
		      \[f'(x) = \odv{x^{\alpha}}{x} = \alpha x^{\alpha - 1}.\]
		      Since $\alpha > 0$, $f'(x) > 0$ for all $x > 0$.
		      Thus, $f$ is strictly monotone increasing.
		\item \textit{Strictly concave:}
		      For concavity, we take the second derivative of $f$ with respect to $x$,
		      \[f''(x) = \odv[2]{x^{\alpha}}{x} = \odv{\alpha x^{\alpha - 1}}{x} = \alpha (\alpha - 1) x^{\alpha - 2}.\]
		      Since $\alpha \in (0,1)$, we have $\alpha - 1 < 0$,
		      thus $f''(x) < 0$ for all $x > 0$.
		      Therefore, we conclude that $f$ is strictly concave.
		\item \textit{Left-endpoint preserved:} We have $f(0) = 0^{\alpha} = 0$.
	\end{itemize}
	Thus, we conclude that power condense satisfies \autoref{def:condense}.
\end{proof}

\subsection{Choice of Maximum Condensation Factor} \label{sec:choice-of-alpha-min}
In this section, we illustrate how we choose the maximum condensation factor $\alpha_{\min}$ in \autoref{eqn:adaptive-condensation}.
Let $x, y \in I$ be the energies of two API functions such that $x < y$,
then their relative ratio after power condense is given by
\[
	R = \frac{y^{\alpha}}{x^{\alpha}} = \left( \frac{y}{x} \right)^{\alpha}.
\]
Let us consider the extreme case after normalization, where $x = \varepsilon$ and $y = 1$, so $\nicefrac{y}{x} = \nicefrac{1}{\varepsilon} = 100$. Since $R = (\nicefrac{y}{x})^{\alpha}$ is increasing in $\alpha$, the \emph{strongest} condensation ($\alpha = \alpha_{\min}$) yields the \emph{smallest} selection ratio between the highest- and lowest-energy API function:
\[
	R_{\alpha_{\min}} = \left( \frac{1}{\varepsilon} \right)^{\alpha_{\min}}.
\]
We choose $\alpha_{\min}$ so that, even under this strongest condensation, the highest-energy API function stays a useful factor more likely to be selected than the lowest-energy one, preventing over-condensation. For example, to keep that factor at $10\times$,
\begin{align*}
	R_{\alpha_{\min}} = \left( \frac{1}{\varepsilon} \right)^{\alpha_{\min}}                & = 10,                                            \\
	\alpha_{\min} \log \left( \frac{1}{\varepsilon} \right) = \log(10) \implies \alpha_{\min} & = \frac{\log(10)}{\log(1/\varepsilon)} = 0.5.
\end{align*}

Through this feedback loop, the API energy distribution dynamically evolves based on the success and novelty of generated $3$-grams, enabling the model to balance correctness reinforcement with exploratory diversity.

\section{Design and Measurement Notes}
\label{sec:design-measurement-notes}

\subsection{Comparison with \textsc{Hopper}}
\label{sec:hopper-comparison}

\lisa builds on Hopper's~\cite{chen2023hopper} distinction between Intra-API (argument validity) and Inter-API (call dependencies). \lisa's advantage lies in its ``decoupling'' strategy, which delegates resolution of Inter-API constraints entirely to \lisaapi. \lisaapi focuses on generating valid and high-quality API sequences (solving the ``reachability'' problem) through LLM reasoning and API $n$-gram feedback before invariant generation. This process drastically reduces the search space complexity, allowing the subsequent stage to concentrate entirely on verifying Intra-API behaviors through invariants, thereby ensuring high-quality assertions.

\subsection{Straight-Line Execution Policy}
\label{sec:straight-line-rationale}

We acknowledge that straight-line programs cannot explicitly encode branches. However, the library under test can still take different internal branches based on internal state, call order, and implicit conditions, even when the generated driver itself is straight-line. This is an intentional trade-off to prioritize correctness and executability while still enabling substantial state evolution and control-flow coverage.

\subsection{Measurement Setup}
\label{sec:measurement-setup}

For \ossfuzz, we used the official fuzzing harnesses from the \ossfuzz project for each library and did not modify them. All coverage metrics reported in the main text are measured at the library level. That is, the coverage only reflects the execution of the target library code, not the fuzzing harness or the test driver. In our implementation, we used \gptmini for the API sequence generation phase (\lisaapi), since this phase focuses on fast, large-scale exploration of API sequences. We used \gptfull for the invariant generation phase (\lisainv), since this phase requires stronger reasoning to produce correct and meaningful invariants.

\section{Metric Definitions and Annotation Protocol}
\label{sec:metric-details}
This appendix supplements the metric definitions in \autoref{sec:metric} with the annotation protocols used for AUVC and HBDR; the core \lisaapi metrics (CSR, ESR, LC, BC) are defined there.

\subsection{AUVC Uniqueness Annotation Protocol} \label{sec:auvc-protocol}
Multiple invariant statements may be considered unique even if they query the same program state. For example, repeated queries over the same zlib stream state following a sequence of API operations are considered distinct. Conversely, syntactically different invariants may be non-unique if they express equivalent conditions (e.g., \texttt{x >= 1} and \texttt{x > 0} for $x \in \mathbb{Z}$). The uniqueness of verifications is determined through manual analysis by two authors with more than three years of C/C++ development experience, who independently inspected and categorized the generated invariants based on their semantic meaning. Their results are discussed and confirmed by a group of senior software researchers and library maintainers. Disagreements were resolved through discussion to reach a consensus.

\subsection{HBDR vs.\ Mutation Score} \label{sec:ibdr-rationale}
While mutation score is a widely adopted metric in evaluating unit test generation techniques~\cite{Andrews2005Is}, synthetic mutants may not fully represent the complexity of real-world software faults~\cite{Papadakis2018Are}. We therefore prefer historical bug re-introduction as a more faithful proxy for real-world fault detection.

\section{API Contract Documentation Details}
\label{sec:api-contract-doc-details}

We construct API contract documentation through two complementary ways: dynamic analysis and natural-language documentation analysis.

\paragraph{Invariants generated by Daikon}
Daikon~\cite{ernst2000daikon} is a dynamic analysis tool designed to detect \emph{likely invariants} in programs by observing their runtime behavior. It reports properties that hold during the observed executions. However, the properties extracted from \emph{observed} executions cannot be directly used to detect functional bugs, since they only capture the snapshot of the current program development stage. Even more critically, if researchers mistakenly treat incorrect or spurious properties as true invariants, correct code with correct logic may be falsely flagged as buggy, leading to misleading or invalid conclusions. Nevertheless, if these Daikon-generated invariants are further manually reviewed and validated, they can still serve as valuable references for understanding program behavior and guiding subsequent analysis, such as refining test oracles or identifying potential correctness properties.

\paragraph{Inferring Invariants From Documentation}
In open-source libraries, developers provide official documentation that describes each API's purpose, parameters, return values, and expected behavior. These expert-crafted manuals include statements such as ``this method is read-only'' or ``this API does not modify internal state'', thereby implicitly defining the behavior invariants of the API. Because the documentation already tells what the API does and what it guarantees, we leverage these descriptions to manually infer invariants (i.e., properties that should always hold when the API is used correctly). For example, if the documentation states that an API is read-only, then a concrete invariant in natural language might be ``no internal state is modified''. Prior work~\cite{Xie2022DocTer} provides evidence that leveraging documentation can effectively support program analysis.

Daikon captures empirical invariants based on runtime behavior, while API documentation provides intended invariants that reflect the developer's design. In practice, we first run Daikon on automatically generated executable API sequences to obtain candidate properties. We use a filtering step to keep invariant candidates that match the official documentation and do not break the program when added as assertions. In our current method, checking if candidates match the documentation requires only a small, one-time human check per API. This step needs no training data or model tuning, and the effort does not grow with the number of tests. This combination improves both the accuracy and coverage of the inferred invariants.

\end{document}